\documentclass[11pt]{article}

\usepackage[margin=1in]{geometry}
\usepackage[utf8]{inputenc} 
\usepackage[T1]{fontenc}    
\usepackage{url}            
\usepackage{booktabs}       
\usepackage{amsfonts}       
\usepackage{nicefrac}       
\usepackage{microtype}      
\usepackage{xcolor}         
\usepackage{natbib}

\usepackage{amsmath, amssymb, amsthm}
\usepackage{thmtools}
\usepackage{caption}

\usepackage{wrapfig}
\usepackage{color}
\usepackage[shortlabels]{enumitem}
\usepackage{algorithm}
\usepackage{xspace}
\usepackage[noend]{algorithmic}

\usepackage{svg}

\usepackage{hyperref}
\hypersetup{
	colorlinks=true,
	linkcolor=blue,
	citecolor=blue,
	filecolor=blue,
	urlcolor=blue
}
\usepackage[nameinlink]{cleveref}
\crefname{section}{\S\!}{\S\!}
\Crefname{section}{\S\!}{\S\!}

\newtheorem{thm}{Theorem}

\newtheorem{lemma}[thm]{Lemma}
\newtheorem{theorem}[thm]{Theorem}
\newtheorem{corollary}[thm]{Corollary}
\newtheorem{assumption}[thm]{Assumption}
\newtheorem{definition}[thm]{Definition}

\crefname{definition}{definition}{definitions}
\Crefname{definition}{Definition}{Definitions}

\newcommand{\algcomment}[1]{\COMMENT{{\color{black!60} #1}}}

\newcommand{\eps}{\varepsilon}

\newcommand{\hX}{\hat{X}}
\newcommand{\tx}{\tilde{x}}
\newcommand{\tX}{\tilde{X}}
\newcommand{\cX}{\mathcal{X}}

\newcommand{\cM}{\mathcal{M}}
\newcommand{\R}{\mathbb{R}}
\newcommand{\Z}{\mathbb{Z}}
\newcommand{\N}{\mathbb{N}}

\newcommand{\cB}{\mathcal{B}}

\newcommand{\cost}{\mathrm{cost}}
\newcommand{\OPT}{\mathrm{opt}}
\newcommand{\opt}{\mathrm{opt}}
\newcommand{\cA}{\mathcal{A}}

\newcommand{\estopt}{\OPT_{\mathrm{est}}}
\newcommand{\Bern}{\mathrm{Ber}}

\newcommand{\Paren}[1]{\left(#1\right)}
\newcommand{\Brac}[1]{\left[#1\right]}
\newcommand{\E}[1]{\mathbb{E}\Brac{#1}}
\newcommand{\rmin}{r_{\min}}

\newcommand{\tZ}{\widetilde{Z}}
\newcommand{\tlambda}{\widetilde{\lambda}}
\newcommand{\bone}{\mathbf{1}}

\newcommand{\atime}{\mathcal{T}}
\newcommand{\aspace}{\mathcal{S}}
\newcommand{\teps}{\tilde{\eps}}
\newcommand{\tdelta}{\tilde{\delta}}

\newcommand{\densealg}{\cA_{\mathrm{dense}}}
\newcommand{\bicrialg}{\cA_{\mathrm{pref}}}
\newcommand{\countingalg}{\cA_{\mathrm{count}}}
\newcommand{\meyersonalg}{\cA_{\mathrm{Meyerson}}}
\newcommand{\dpcoreset}{\cA_{\mathrm{dpcoreset}}}

\newcommand{\polylog}[1]{\mathrm{polylog}\Paren{#1}}
\newcommand{\tO}[1]{\tilde{O}\Paren{#1}}
\newcommand{\dist}{\mathrm{dist}}
\newcommand{\mt}{\mathrm{mt}}

\newcommand{\Expt}{\mathsf{Game}}

\allowdisplaybreaks

\renewcommand{\setminus}{\smallsetminus}

\newcommand{\variantalgo}{\Cref{alg:meyerson-main}$'$\xspace}

\title{Online Differentially Private Consistent Clustering}

\author{%
Edith Cohen\thanks{Google Research and Tel Aviv University. \texttt{edith@cohenwang.com}}
  \and
Vadym Doroshenko\thanks{Google. \texttt{dvadym@google.com}}
\and
Badih Ghazi\thanks{Google Deepmind. \texttt{badihghazi@gmail.com}}
\and
Pritish Kamath\thanks{Google Research. \texttt{pritish@alum.mit.edu}, \texttt{alexanderknop@google.com}, \texttt{ravi.k53@gmail.com}, \texttt{ethanleeman@google.com}, \texttt{pasin@google.com}, \texttt{adamsealfon@google.com}, \texttt{marikaswanberg@google.com}}
\and
Alexander Knop\footnotemark[4]
\and
Ravi Kumar\footnotemark[4]
\and
Ethan Leeman\footnotemark[4]
\and
Pasin Manurangsi\footnotemark[4]
\and
Adam Sealfon\footnotemark[4]
\and
Marika Swanberg\footnotemark[4]
}

\date{\today}

\begin{document}

\maketitle

\begin{abstract}
    We study differentially private (DP) $k$-means 
    and $k$-median clustering in the online streaming setting.  In this model, points arrive sequentially, 
    and at each time step, we need to output a set of $k$ centers that optimizes the clustering 
    objective for all points seen so far. 
    We give a generic reduction that transforms the (sensitive) input stream into a private stream, which is a {\em semi-coreset} of the input stream.
    This implies that any (non-private) online clustering algorithm, run as a post-processing step, can achieve good utility for the original clustering objective.
    Our algorithm matches or improves upon the approximation ratio, 
    space usage, and running time of existing algorithms \citep{EMZ23,TourHS24}. 
    A key aspect of our reduction is that it inherits desirable properties of the underlying 
    non-private clustering algorithm, such as {\em consistency}~\citep{LattanziV17}---a property
    not satisfied by previous DP algorithms.
\end{abstract}

\section{Introduction}

Differential Privacy (DP)~\citep{DworkKMMN06,dwork2006calibrating} 
has emerged as the mathematically rigorous, widely used gold standard for privacy-preserving large-scale data analysis. 
By injecting carefully calibrated statistical noise into computations, DP provides a formal, quantifiable guarantee that limits the information an adversary can infer about any single data point, making it a foundational tool for government agencies and technology companies alike (e.g.,  \citet{abowd2018us,erlingsson2014rappor}).

Within the domain of large-scale data analysis, clustering remains one of the most basic unsupervised learning tasks. 
In this work, we focus on the \emph{$(k, p)$-clustering objective}: 
given a dataset $X = (x_1, \dots, x_n) \in (\R^d)^n$, the goal is to output a set $C$ of $k$ centers that minimizes $\cost^p_X(C) := \sum_{i \in [n]} \dist(x_i, C)^p$, where $\dist(x_i, C)$ denotes the (Euclidean) distance from $x_i$ to its closest point in $C$. 
When $p = 1$ (resp., $2$), this corresponds to the \emph{$k$-median} (resp., \emph{$k$-means}) clustering, which are among the most well-studied objectives. 
Since finding the exact minimizer is NP-hard, clustering algorithms focus on providing a small multiplicative 
\emph{approximation ratio}~\citep{Matousek00,KanungoMNPSW04,HarPeledM04,AwasthiCKS15,LeeSW17,CohenAddadS19,FriggstadRS19,AhmadianNSW20}. We assume throughout that $p \geq 1$ is a fixed constant and we will omit the dependency on $p$ in notations and bounds for brevity.

Traditionally, clustering algorithms were designed for static datasets, where $x_1, \dots, x_n$ are processed
simultaneously. 
However, in modern applications, continuous data generation and collection are prevalent. 
Online clustering is a natural task that arises in analyzing dynamically changing data streams. 
In the streaming (a.k.a. continuous release) model, the data points $x_1, \dots, x_n$ arrive sequentially and after the arrival of each $x_t$, the algorithm must return a set $C_t$ of $k$ centers for all the data points seen so far, i.e., $x_1, \dots, x_t$. 
Non-private algorithms with small approximation ratios and low space complexity are known in this model, e.g.,~\citep{GuhaMMO00,CharikarOP03,HarPeledM04}.

\paragraph{DP Clustering.}
Given the significance of clustering, a vast body of literature has explored DP clustering algorithms~\citep{blum2005practical,GuptaLMRT10,BalcanDLMZ17,feldman2017coresets,NissimSV16,NissimS18,huang2018optimal,Stemmer20,GhaziKM20,ChangGKM21,GhaziHK0MNV23,TourHS24}. 
In particular, significant progress has been made in the static setting, where efficient DP clustering algorithms with constant approximation ratios and nearly-optimal additive errors are known~\citep{GhaziKM20,TourHS24}.
Unfortunately, the same cannot be said in the streaming model. 
Specifically, while \citet{TourHS24} proposed a DP 
online $(k, p)$-clustering algorithm with nearly optimal approximation ratios and additive errors, their algorithm is inefficient, requiring super-polynomial time and space. 
Meanwhile, \citet{EMZ23} proposed an efficient algorithm---based on the merge-and-reduce framework for coresets~\citep{AgarwalHV04,FeldmanSS20}---but with suboptimal additive errors in terms of the dimension $d$. 
Furthermore, neither algorithm is designed to preserve desirable properties (such as consistency) of their underlying 
non-private counterpart, as explained next.

\paragraph{Consistency in Online Clustering.}
While online clustering algorithms must adapt to changing data, this can cause a volatile output sequence. 
For example, if an algorithm blindly recalculates cluster centers based on $x_1, \dots, x_t$ for every $t \in [n]$, the cluster centers can vary wildly even
across consecutive time steps. 
The need for a temporally stable output sequence motivated
the definition of consistency~\citep{LattanziV17}: 
An output sequence $(C_1, \dots, C_n)$ is said to be \emph{$m$-consistent} if $\sum_{t \in [n]} |C_t \setminus C_{t - 1}| \leq m$ where we use the convention $C_0 = \emptyset$, i.e., the total number of centers removed/added throughout the entire sequence is at most $m = m(n)$.  Consistency ensures that the clusters do not change unpredictably at every time step. 
Maintaining consistency is preferred because temporal  instability can be undesirable in many real-world scenarios. 

Since the notion was formalized, many consistent clustering algorithms have been proposed in the literature~\citep{LattanziV17,GuoKLX21,FichtenbergerLN21,ChanJWZ25}. In particular, \citet{FichtenbergerLN21} devised a non-private algorithm for clustering in arbitrary metric space that achieves $O(1)$-approximation 
and $\tilde{O}(k)$-consistency (which is known to be nearly optimal~\citep{LattanziV17}).
Unfortunately, as mentioned earlier, neither of the previous DP online clustering algorithms~\citep{TourHS24,EMZ23} provides non-trivial guarantees on the consistency of their output sequence.


\subsection{Our Contributions}

Our primary contribution is a novel generic reduction
that transforms a sensitive stream to a privatized stream that approximately preserves the clustering objective value on any input prefix.  As a key corollary, we obtain the first DP algorithm that has a nearly optimal consistency guarantee. Furthermore, our algorithm obtains a constant approximation ratio, nearly optimal additive error in terms of the dimension $d$, and almost linear time and sublinear space for $d, k = n^{o(1)}$, as stated below.

\begin{theorem}[Informal version of \Cref{thm:dp-consistent-main}]
There is an $(\eps, \delta)$-DP online $(k, p)$-clustering algorithm that, given a stream $X = (x_1, \dots, x_n)$ of points, outputs a stream $C = (C_1, \dots, C_n)$ of center multisets such that for each $t \in [n]$, $|C_t| = k$ and with high probability, $C_t$ is an approximate solution for $x_1, \ldots, x_t$ with constant approximation ratio and $\frac{\sqrt{d}}{\eps} \cdot \Paren{k \log(nd/\delta)}^{O(1)}$-additive error. Moreover, the output is $\Paren{k \cdot \polylog{nd}}$-consistent.

Furthermore, when $k, d, \log(1/\delta) \leq n^{o(1)}$, the algorithm runs in $n^{1 + o(1)}$ time over the entire stream and uses $n^{o(1)}$ space.
\end{theorem}

As mentioned above, both the dependency of $\sqrt{d} \cdot \polylog{d}$ in the additive error, and the consistency bound $k \cdot \polylog{nd}$ are nearly optimal, matching the known lower bounds~\citep{PeterTU24,LattanziV17} up to polylogarithmic factors. A detailed comparison between our algorithm and those in \cite{TourHS24,EMZ23} can be found in \Cref{tab:dp_clustering_comparison}.

\begin{table}[t]
    \centering
    \caption{
    Comparison of online $(\eps,\delta)$-DP clustering algorithms. For brevity, we hide $\polylog{nd/\delta}$ terms in all bounds and assume $\eps = \Theta(1)$. 
    All algorithms have constant approximation ratio. 
    While \cite{TourHS24} do not specify the complexity of their algorithm, the bottleneck in their algorithm is the use of an $O(1/n)$-net in $\R^{O(\log k)}$, which requires $n^{\Theta(\log k)}$ time and space. 
    We note however that their algorithm supports stream with both insertions and deletions, whereas \cite{EMZ23} and ours only support insertions.  Note that $o_{k, d}(n)$ consistency cannot be achieved for streams with both
    insertions and deletions even non-privately and when $O_{k,d}(1)$ additive error is allowed.} 
    \label{tab:dp_clustering_comparison}
    \begin{tabular}{l c c c c}
        \toprule
        \textbf{Algorithm} & \textbf{Additive Err.} & \textbf{Space} & \textbf{Time} & \textbf{Consistency} \\
        \midrule
        \cite{TourHS24} & $k^{O(1)} \sqrt{d}$ & $n^{O(\log k)} d^{O(1)}$ & $n^{O(\log k)} d^{O(1)}$ & --- \\
        \cite{EMZ23} & $k^{2.5} d^{3.51}$ & $(kd)^{O(1)}$ & $n (d k)^{O(1)}$ & --- \\
        This paper (\Cref{thm:dp-consistent-main}) & $k^{2.5} \sqrt{d}$ & $(kd)^{O(1)}$ & $n d k^2 + (dk)^{O(1)}$ & $k$ \\
        \bottomrule
    \end{tabular}
\end{table}
We now state our generic reduction. 
Recall that for a given multiset $X$, a semi-coreset $\tilde{X}$ is
its ``compressed'' version such that for any set of cluster centers,
the clustering objective on $X$ and $\tilde{X}$ are close to each other, up to a multiplicative error in the approximation factor and an additive error in the cost 
(see \Cref{def:semi-coreset}.) Our reduction is an online algorithm that produces an output stream $\tX = (\tx_1, \dots, \tx_n)$ of points%
\footnote{
As stated our algorithm outputs a stream of \emph{multisets} of points rather than individual points, i.e., multiple points may be released in each step. See \Cref{thm:main-semi-coreset} for the formal statement. We note that it is possible to modify the reduction so that at most one point is released in each step, but this distinction is not important in downstream applications.
} 
such that, for any $t \in [n]$, the output prefix $\tx_1, \dots, \tx_t$ is a semi-coreset of the input prefix $x_1, \dots, x_t$, as stated below; furthermore, our algorithm is DP.  
\begin{theorem}[Informal version of \Cref{thm:main-semi-coreset}] \label{thm:coreset-informal}
    There is an $(\eps, \delta)$-DP algorithm that, given a stream $X = (x_1, \dots, x_n)$, 
    produces another stream $\tX = (\tx_1, \dots, \tx_n)$ such that, with high probability,
    $\{\tx_1, \dots, \tx_t\}$ is a semi-coreset for $\{x_1, \dots, x_t\}$ with constant multiplicative error and $\frac{\sqrt{d}}{\eps} \cdot \Paren{k \log(nd/\delta)}^{O(1)}$-additive error 
    for all $t \in [n]$.
    
    Furthermore, when $k, d, \log(1/\delta) \leq n^{o(1)}$, the algorithm runs in $n^{1 + o(1)}$ time over the entire stream and uses $n^{o(1)}$ space.
\end{theorem}

Since DP is preserved under post-processing, we can run any \emph{non-private} online clustering algorithm on the output $\tX$. Indeed, our DP consistent clustering algorithm (\Cref{thm:dp-consistent-main}) simply runs the (non-private) algorithm by~\cite{FichtenbergerLN21} as a post-processing step on $\tX$.

\subsection{Technical Overview}
\label{subsec:overview}


\paragraph{(Non-Private) Meyerson's Sketch.}
The starting point of our work is the so-called \emph{Meyerson's sketch}~\citep{meyerson2001online}. To describe the sketch, it is simpler to consider the setting where, instead of keeping a semi-coreset, we only aim for an \emph{end-stream bi-criteria} approximation. Here, we need to output a set $S$ of points at the end of the input stream $X = (x_1, \dots, x_n)$ such that
(i) $\cost^p_X(S)$ is at most $O(1)$ times the optimum and 
(ii) $|S| \leq O(k \log n)$.
We use the nomenclature ``bi-criteria'' since we allow $S$ to be (slightly) larger than $k$.  Standard techniques exist  for turning such an end-stream bi-criteria approximation into an end-stream 
semi-coreset, even with DP; see \Cref{subsec:bicriteria-to-coreset} for details.
\begin{wrapfigure}{r}{0.41\textwidth}
\small
\hrule height 0.8pt
{
\captionsetup{type=algorithm,labelfont=bf,textfont=normalfont}
\captionof{algorithm}{Meyerson's Sketch \citep{meyerson2001online}}
}
\label{alg:meyerson-nonprivate}
\vspace{-0.6em}
\hrule
\begin{algorithmic}[1]

\makeatletter
\renewcommand{\ENDFOR}{\end{ALC@for}}
\makeatother

\REQUIRE Stream $X = (x_1, \dots, x_n)$
\item[\textbf{Parameters:}] 
 Estimated optimum $\estopt$
\STATE $S = \emptyset$
\FOR{$t = 1, \dots, n$}
\STATE $\lambda_t \gets \min\left\{1, \frac{k \log n}{\estopt} \cdot \dist(x_t, S)^p \right\}$ \\
\algcomment{If $S = \emptyset$, let $\dist(x_t, S) = \infty$}
\STATE Sample $Z_t \sim \Bern(\lambda_t)$
\IF{$Z_t = 1$} 
\STATE $S \gets S \cup \{x_t\}$ 
\ENDIF
\ENDFOR
\STATE Output $S$
\end{algorithmic}
\hrule
\end{wrapfigure}
We now describe the Meyerson sketch.  Initially, $S = \emptyset$. After a point $x_t$ arrives, we compute its distance to its closest point in $S$, denoted by $\dist(x_t, S)$. Then, with probability $\lambda_t$, we add it to $S$; otherwise, we simply discard it.  Here, $\lambda_t$ is set based on $\dist(x_t, S)$ so that the larger the distance, the more likely it is to be included in $S$. The full description is in \Cref{alg:meyerson-nonprivate}, which also uses an estimate (i.e.,  ``guess'') of the optimum value $\estopt$. Standard analysis shows that if $\estopt$ is within a constant factor of the true optimum, then the output $S$ is a bi-criteria approximation with constant probability~\citep{meyerson2001online,LattanziV17,FichtenbergerLN21}.

\paragraph{Privacy via DP Dense Balls.} 
A naive attempt to make Meyerson's sketch (in fact, any sampling-based sketch) DP faces the following issue: The output contains (raw) input points, rendering it blatantly non-private.  
To overcome this challenge, we do \emph{not} return the set $S$ directly. Instead, we maintain another 
set $R$, which serves as a proxy for $S$ but is DP, and use the distance $\dist(x_t, R)$ instead of $\dist(x_t, S)$ in each step. Our main observation is that the analysis of the Meyerson's sketch is quite flexible: as long as $R$ contains a ``representative point'' in every ``dense ball'' of $S$, it is already sufficient to obtain a bi-criteria approximation (with some small additive error) while still ensuring that $S$ remains not ``too large''. We formulate this requirement for the proxy set $R$
as the \emph{dense ball} problem (\Cref{def:dense-ball}) and show that known algorithms from the literature~\citep{NissimSV16,NissimS18,GhaziKM20} can be adapted to solve this problem privately. 

With this formulation, we then compute $R$ by running a DP dense ball algorithm on $S$. However, doing so at every time step $t \in [n]$ would mean we need to apply the DP composition theorem on the privacy budget across all the steps, resulting in excessive error (that is polynomial in $n$). To overcome this, we observe that since $S$ remains small, it suffices to invoke the DP dense ball algorithm only when $S$ is updated.  However, we cannot naively perform this check since $S$ contains sensitive data; hence, we use the \emph{continual counting}  mechanism~\citep{DworkNPR10,ChanSS11} to provide a DP  estimate of $|S|$ after every time step. We then invoke
the DP dense ball algorithm on $S$ only when the (privately estimated) size of $S$ is sufficiently larger than what it was when the DP dense ball was last invoked.  This concludes our high-level technical overview.

Sampling is the most popular method for finding a (semi-)coreset for (non-private) clustering problems and beyond (e.g.,  \cite{Chen09,FeldmanL11,FeldmanSS20,CohenAddadSS21}).  We believe our work is the first to naturally integrate sampling-based algorithms with DP, and hope that our techniques will be useful for future research on this topic.

\section{Preliminaries}
\label{sec:prelim}

For any metric space $\cM = (U, \dist)$, 
%
let $\cB(u, r)$ denote the \emph{ball} of radius $r$ centered at $u$, i.e., 
$\cB(u, r) := \{u' \in U \mid \dist(u, u') \leq r\}$.
For any $S \subseteq U$, let $\dist(u, S) = \inf_{u' \in S} \dist(u, u')$.

Throughout this work, we assume that we work with the Euclidean space\footnote{Our DP Meyerson's sketch (\Cref{alg:meyerson-main}) works for any metric space as long as there is an efficient DP dense ball algorithm for that space.
However, we are only aware of existing work for the Euclidean space.} with unit diameter\footnote{The bounded diameter assumption is necessary to obtain finite additive error~\citep{NissimSV16}.
By scaling, our algorithm works for any diameter bound $\Delta$, incurring a multiplicative factor of $\Delta^p$ in the additive error.}.
%
\begin{assumption}
$\cM$ is the Euclidean space in $\R^d$ with $U = \cB\Paren{0, \frac{1}{2}}$.
\end{assumption}


A \emph{stream} (of points)
is a sequence $X = (x_1, \dots, x_n)$, where each $x_t$ belongs
to $U \cup \{\perp\}$.
Here $\perp$ should be interpreted as ``no point''; for convenience, we abuse the notation and write $\dist(\perp, S) = 0$ for any set $S \subseteq U$. For every $t \in [n]$, we write $X_{\leq t}$ to denote the multiset $\{x_i \mid i \in [t], x_i \ne \perp\}$.

We also define a \emph{multiset stream} $Y = (y_1, \dots, y_n)$, where each $y_i$ is a multiset of $U$.
Again, we write $Y_{\leq t}$ to denote the multiset $\bigcup_{i \in [t]} y_i$. The \emph{size} of $Y$ is $|Y_{\leq n}|$.

An \emph{online algorithm} $\cA$ is an algorithm that receives the points in the stream $X = (x_1, \dots, x_n)$ one point per time step, and, immediately after receiving $x_t$ at time step $t$, it must output $o_t$. The final output released over the execution is $\cA(\cX) = (o_1, \dots, o_n)$. 

\subsection{\boldmath (Online) \texorpdfstring{$(k, p)$}{(k, p)}-Clustering}
\label{subsec:prelim-online-clustering}

\paragraph{Static Clustering.}
For a given multiset $X = \{x_1, \dots, x_n\} \subseteq U$ of input points, the cost for a set $C \subseteq U$ of centers is given by $\cost^p_X(C) := \sum_{i \in [n]} \dist(x_i, C)^p$, where $p \geq 1$ is a constant.
In the \emph{$(k, p)$-clustering problem}, the goal is to find a set $C$ of $k$ centers that minimizes $\cost^p_X(C)$.
It will be convenient to define the optimum as $\OPT^{k,p}_X := \min_{C \in \binom{U}{k}} \cost^p_X(C)$.

\paragraph{Online Clustering.}
For a given stream $X = (x_1, \dots, x_n)$, a
multiset stream $C = (C_1, \dots, C_n)$ is an \emph{$(\alpha, \beta)$-approximation} if
\begin{align*}
    \cost^p_{X_{\leq t}}(C_t) \leq \alpha \cdot \OPT^{k,p}_{X_{\leq t}} + \beta & &\forall t \in [n].
\end{align*}

An \emph{online $(k, p)$-clustering algorithm} takes in an input stream $X = (x_1, \dots, x_n)$ and outputs a multiset stream $C = (C_1, \dots, C_n)$ such that $|C_t| = k$ for all $t\in [n]$.
We say that it is an \emph{$(\alpha, \beta)$-approximation} algorithm for online $k$-clustering if for every input stream, the output
stream of the algorithm 
is an $(\alpha, \beta)$-approximation
(with probability $1 - \gamma$).

\paragraph{Prefix-Approximation.}
For our Meyerson's sketch construction, it will be helpful to define the following notion.
We say that a multiset stream 
$Y = (y_1, \dots, y_n)$ is an \emph{$(\alpha, \beta)$-prefix-approximation} (w.r.t. $(k, p)$-clustering) of a stream $X = (x_1, \dots, x_n)$ if 
\begin{align*}
    \sum_{t \in [n]} \dist(x_t, Y_{\leq t})^p \leq \alpha \cdot \OPT^{k, p}_{X_{\leq n}} + \beta.
\end{align*}
Note that prefix-approximation is stronger than requiring $Y_{\leq n}$ to be an $(\alpha, \beta)$-approximation of $X_{\leq n}$, since
prefix-approximation computes the cost using $\dist(x_t, Y_{\leq t})$.
In other words, we only allow $x_t$ to be connected to the point closest to it only up until that point in the stream (or, equivalently, the points in $y_{t + 1}, \dots, y_n$ cannot be used).
On the other hand, we do not restrict the size of $Y$ (which will be stated explicitly in the theorem statements, e.g.,  \Cref{thm:combined-bicriteria-high-prob}).

\subsubsection{Semi-Coreset}

Next, we recall the definition of a {\em semi-coreset} \citep{EMZ23}. Unlike a standard coreset~\citep{HarPeledM04},  which provides a purely multiplicative approximation,  a
semi-coreset allows additional additive error terms.

\begin{definition}[Semi-Coreset~\citep{EMZ23}] \label{def:semi-coreset}
    A multiset $\tX$ is a \emph{$(\kappa, \eta_1, \eta_2)$-semi-coreset} for $(k, p)$-clustering of $X$ if, for all $C \in \binom{U}{k}$, we have
    \[\textstyle
    \frac{1}{\kappa} \cdot \cost^p_X(C) - \eta_1 \cdot \OPT^{k, p}_X - \eta_2
    ~\leq~ \cost^p_{\tX}(C) ~\leq~
    {\kappa} \cdot \cost^p_X(C) + \eta_2.
    \]
\end{definition}
When $k, p$ are clear from context, we often drop them from the notation for brevity.

We say that a multiset stream $\tX = (\tx_1, \dots, \tx_n)$ is a $(\kappa, \eta_1, \eta_2)$-semi-coreset for another multiset stream $X = (x_1, \dots, x_n)$ if $\tX_{\leq t}$ is a $(\kappa, \eta_1, \eta_2)$-semi-coreset of $X_{\le t}$ for all $t \in [n]$.

While the semi-coreset definition is more relaxed than the original coreset definition, it still maintains the most crucial property of coreset: 
Any good approximation for the semi-coreset is also a
good approximation for the original dataset.
This is stated more precisely below\footnote{Missing proofs in this section (including  \Cref{lem:coreset-approx}) and additional preliminaries can be found in Appendix~\ref{app:prelim}.}.

\begin{lemma}[Approximation from Semi-Coreset] \label{lem:coreset-approx}
    Let $X$ be a multiset of points and let $\tX$ be a $(\kappa, \eta_1, \eta_2)$-semi-coreset for $X$.
    If a set $C \in \binom{U}{k}$ of centers is an $(\alpha, \beta)$-approximation for $\tX$, then $C$ is an $(\alpha', \beta')$-approximation for $X$, where 
    $\alpha' = \alpha \kappa^2 + \eta_1 \kappa$ and $\beta' = \kappa(\alpha + 1)\eta_2 + \kappa\beta$.
\end{lemma}

\subsection{Consistent Clustering}
\label{subsec:prelim-consistent}


For an online $(k, p)$-clustering algorithm, we say that an output stream $C = (C_1, \dots, C_n)$ of centers is \emph{$m$-consistent} if $\sum_{t \in [n]} |C_t \setminus C_{t - 1}| \leq m$, where $m = m(n)$.  
We say that the algorithm is $m$-consistent if its output is always $m$-consistent.
We use the following result by \cite{FichtenbergerLN21}, which yields a nearly optimal consistency bound in the non-private setting.



\begin{theorem}[Consistent $k$-Clustering with Constant Additive Error~\citep{FichtenbergerLN21}] \label{thm:consistent-additive}
There is an online $(k, p)$-clustering $\Paren{k \cdot \polylog{nd}}$-consistent algorithm whose output is an $(O(1), O(1))$-approximation with high probability.
Moreover, its time and space complexity is polynomial in the number of distinct points in the input stream and the number of dimensions $d$.
\end{theorem}

We remark that the above result is slightly different than the one originally stated by \cite{FichtenbergerLN21}, whose bound depends on the ``aspect ratio'' of the metric space.
However, there is a simple reduction between the two versions.
See \Cref{subsec:consistent-snapping} for more details.

\subsection{Dense Balls}

As mentioned in \Cref{subsec:overview}, a central problem in our algorithm is the dense ball problem, which is parameterized by $\rmin$, the smallest target radius, and $\tau$, the target number of points per ball.
The input to the problem is a multiset $S$
of points and the output is another multiset $O$.
(Let $\digamma = 3$ be a factor that is fixed throughout.) The problem is formalized as follows.

\begin{definition}[Dense Ball] \label{def:dense-ball}
We say that an algorithm solves the 
\emph{$(r_{\min}, \tau)$-dense ball} (with probability $1 - \gamma$) if the following \emph{coverage constraint} holds for any input multiset $S$ of at most $N$ points (with probability $1 - \gamma$): 
$\dist\Paren{u, O} \leq \digamma \cdot r$ for all $u \in U, r \geq r_{\min}$ such that $|S \cap \cB(u, r)| \ge \tau$.

Moreover, we say that the algorithm has \emph{size blow-up $\Upsilon$ with size capacity $N$} if, for any input $S$ of size at most $N$, the output $O$ (always) satisfies $|O| \leq \Upsilon \cdot \frac{N}{\tau}$.
\end{definition}

\subsection{DP and Continual Counting}
\label{sec:prelim-dp}

We quickly recall the definition of DP. For online clustering algorithms, two input streams $X = (x_1, \dots, x_n), X' = (x'_1, \dots, x'_n)$ are \emph{(add-remove) neighbors}%
\footnote{Note that one may consider the \emph{substitution neighbors}, in which we can change $x_{i^*}$ to any other point $x'_{i^*}$ (not necessarily $\perp$).
By standard {\em group privacy}, our algorithm also applies in this setting.} iff there exists $i^* \in [n]$ such that $x_i = x'_i$ for all $i \ne i^*$, and $\perp \in \{x_{i^*}, x'_{i^*}\}$. 
In words, one stream can be obtained from the other by changing a single point to $\perp$.

\begin{definition}[Differential Privacy (DP) \citep{DworkKMMN06,dwork2006calibrating}] \label{def:nonadaptivedp}
    An online algorithm $\cA$ that takes in a stream is said to be \emph{$(\eps, \delta)$-DP} if, for all sets $S$ of possible outputs and all pairs $X, X'$ of neighboring input streams, it holds that $\Pr[\cA(X) \in S] \leq e^{\eps} \cdot \Pr[\cA(X') \in S] + \delta$.
\end{definition}
The case $\delta = 0$ is known as \emph{pure-DP} (written as $\eps$-DP) and
$\delta > 0$ is known as \emph{approximate-DP}.

Throughout this work, we will assume that $\eps \in (0, 1)$ and $\delta \in (0, 1/2)$.

We remark that the above model is the so-called \emph{non-adaptive} (aka \emph{oblivious}) model where the adversary fixes the neighboring input streams at the start of the algorithm. Our privacy proof however works even with the stronger \emph{adaptive} (aka \emph{adversarial}) model where the adversary may choose to adapt the streams based on the outputs so far. Due to the extra notations required, we defer the formal definition of this adaptive model to the appendix. However, we stress that our utility proof assumes the oblivious model, and it remains an interesting question to extend the utility guarantees to the adversarial setting. (See \Cref{sec:conclusion} for more discussion.)


\paragraph{Continual Counting.} In DP continual counting~\citep{DworkNPR10,ChanSS11},
 the input stream is $Y = (y_1, \dots, y_n) \in \{0, 1\}^n$.
The goal is to output a private estimate $e_t$ of $\sum_{t' \in [t]} y_{t'}$ after receiving $y_t$.
We say that the estimate is 
\emph{$\alpha$-accurate}
if $\left|e_t - \sum_{t' \in [t]} y_{t'}\right| \leq \alpha$ for all $t \in [n]$

\section{Private Semi-Coreset Construction} 
\label{sec:dp-coreset-main}

Our main DP online semi-coreset construction (a formal version of \Cref{thm:coreset-informal}) is stated below.

\begin{theorem}[DP Online Semi-Coreset] \label{thm:main-semi-coreset}
There is an online $(\eps, \delta)$-DP algorithm whose output is an $(O(1), O(1), \beta)$-semi-coreset 
of the input with probability $1 - \gamma$, and the algorithm runs in time $\tO{ndk^2 \cdot \log(1/\gamma)} + \Lambda^{O(1)}$ and uses space $\Lambda^{O(1)}$ where
\begin{itemize}
    \item $\beta = O\Paren{\frac{d \cdot k^3}{\eps} \cdot \polylog{\frac{n d}{\eps \gamma}} }$ and $\Lambda = \frac{kd \log(n/\gamma)}{\eps}$ for pure-DP.
    \item $\beta = O\Paren{\frac{\sqrt{d} \cdot k^{2.5}}{\eps} \cdot \polylog{\frac{n d}{\eps \delta \gamma}} }$ and $\Lambda = \frac{kd \log\Paren{ \frac{n}{\gamma\delta} } }{\eps}$ for approximate-DP.
\end{itemize}
\end{theorem}

\subsection{From Dense Ball to Prefix-Approximation: Private Meyerson's Sketch}
\label{sec:dp-meyerson-main}

The main ingredient for \Cref{thm:main-semi-coreset} is a private version of Meyerson's sketch (\Cref{alg:meyerson-nonprivate}).  As in the non-private version, we maintain
a set $S_t$ of sampled input points, but instead of using these directly we periodically feed them to the DP dense ball algorithm to get $R_t$, which we actually output. For notational convenience, let $R_0 = \{0\}$ contain the origin, and write $R_{\leq t} = \bigcup_{i=0}^t R_i$ to include $R_0$ as well. Moreover, we use the convention  $\dist(\perp, R_{\leq t}) = 0$ so that $x_t = \perp$ is never included in the set $S_t$.  Recall that $\estopt$ is a guess of the optimum value.  Let $r_{\min} := \Paren{\frac{\estopt}{n}}^{1/p}$ and $L := \left\lceil \log_2\Paren{\frac{1}{r_{\min}}} \right\rceil$.  \Cref{alg:meyerson-main}  follows the outline from \Cref{subsec:overview} except for additional checks to ensure that the size of $S_t$ and the number of calls to the DP dense ball algorithm are bounded.  

\begin{algorithm}[t]
\small
\caption{Private Meyerson's Sketch: $\meyersonalg(\estopt, \Psi)$}
\label{alg:meyerson-main}
\begin{algorithmic}[1]

\makeatletter
\renewcommand{\ENDFOR}{\end{ALC@for}}
\makeatother

\REQUIRE Stream $X = (x_1, \dots, x_n)$
\item[\textbf{Parameters:}] Target sketch size $\Psi$; Dense ball algorithm $\densealg(\tau, \Psi, \Upsilon)$ with threshold $\tau$, size capacity $\Psi \cdot \tau$ and size blow-up $\Upsilon$; Continual counting algorithm $\countingalg$; 
Estimated optimum $\estopt$
\STATE $(R_0, S_0) \gets (\{0\}, \emptyset)$
\STATE $j = 1$ \hfill \algcomment{Counter for when to call $\densealg$}
\FOR{$t = 1, \dots, n$}
\STATE \algcomment{Do the following upon receiving $x_t$}
\STATE $\lambda_t \gets \min\left\{1, \frac{\tau k L}{\estopt} \cdot \dist(x_t, R_{\leq t-1})^p \right\}$ \label{line:closest-point}
\hfill \algcomment{Probability of adding $x_t$ to $S_t$}
\STATE $S_t \gets S_{t - 1}$
\STATE Sample $Z_t \sim \Bern(\lambda_t)$ \hfill \algcomment{Should $x_t$ be added to $S_t$?} 
\IF{$Z_t = 1$ and $|S_t| < \Psi \cdot \tau$} \label{line:size-check-s}
\STATE $S_t \gets S_t \cup \{x_t\}$ 
\ENDIF
\STATE Feed $Z_t$ to $\countingalg$ and let $e_t$ be the output
\IF{$e_t \geq j \cdot \tau$ and $j \leq \Psi$} \label{line:num-call-dense-check}
\STATE $R_t \gets \densealg(S_t)$
\STATE $j \gets j + 1$
\ELSE
\STATE $R_t \gets \emptyset$
\ENDIF
\STATE Output $R_t$
\ENDFOR
\end{algorithmic}
\end{algorithm}

A generic version of the approximation
guarantee of \Cref{alg:meyerson-main} is given below.
\begin{lemma}[Prefix-Approximation Guarantee] \label{lem:bicriteria-approx}
Assume that
\begin{itemize}
\item $\densealg(\tau, \Psi, \Upsilon)$ solves $\Paren{r_{\min}, \tau}$-dense ball and has size blow-up $\Upsilon$ with size capacity $\Psi \cdot \tau$.
\item $\countingalg$ solves continual counting with accuracy $\tau/2$. 
\item $\textstyle\Psi \geq D \cdot \Paren{kL \cdot \Paren{1 + \frac{\OPT^{k, p}_{X_{\leq n}}}{\estopt}}}$ for some sufficiently large constant $D$.
\end{itemize} 
Then, the output stream $(R_1, \dots, R_n)$ is an $\Paren{O(1), O\Paren{\estopt + \tau k L}}$-prefix-approximation of $(x_1, \dots, x_n)$ with probability $1/2$. 
Furthermore, the output stream has size $|R_{\leq n}| \leq \Psi^2 \cdot \Upsilon$.

\end{lemma}

\begin{proof}
The size bound immediately follows due to the checks $j \leq \Psi$ on line~\ref{line:num-call-dense-check}, $|S_t| \leq \Psi \cdot \tau$ on line~\ref{line:size-check-s} and the size blow-up guarantee of $\densealg$.

Next, we will prove the prefix-approximation claim.
%
To do this, it will be more convenient to analyze a modified version of \Cref{alg:meyerson-main}, denoted
\variantalgo, removing the
checks that $|S_t| < \Psi \cdot \tau$ on line~\ref{line:size-check-s} and that $j \leq \Psi$ on line~\ref{line:num-call-dense-check}.
For \variantalgo, we will then show the following:
\begin{align} \label{eq:bicri-approx-expectation}
\textstyle
\E{\sum_{t \in [n]} \dist(x_t, R_{\leq t - 1})^p} \leq O\Paren{\OPT_{X_{\leq n}}^{k, p} + \estopt + \tau k L},
\end{align}
\begin{align} \label{eq:size-s-expectation}
\textstyle
\E{|S_{\leq n}|} \leq \frac{\Psi \cdot \tau}{8}.
\end{align}
Before we prove these, we show how they complete our main proof. By Markov's inequality and a union bound, \eqref{eq:bicri-approx-expectation}, \eqref{eq:size-s-expectation} imply that, w.p. $3/4$, both the prefix-approximation guarantee and $|S_t| \leq \Psi \cdot \tau$ holds in \variantalgo. When the latter holds (and since $\countingalg$ is $(\tau/2)$-accurate), both checks on lines~\ref{line:size-check-s} and~\ref{line:num-call-dense-check} in \Cref{alg:meyerson-main} pass; this means that \Cref{alg:meyerson-main} produces the same output $(R_1, \dots, R_t)$ as \variantalgo.  Thus, the prefix-approximation guarantee holds with probability $3/4$ for \Cref{alg:meyerson-main}.

To prove \eqref{eq:bicri-approx-expectation} and \eqref{eq:size-s-expectation}, let $C^* = \{c^*_1, \dots, c^*_k\}$ denote the optimal centers such that $\cost^p_{X_{\leq n}}(C^*) = \OPT^{k, p}_{X_{\leq n}}$. 
For each $i \in [k]$, let $P_i = \{ t \in [n] \mid x_t \mbox{ is closest to } c_i^*, 
\mbox{ ties broken arbitrarily }\}$.

For every $\ell = 1, \dots, L$, define the
partitions
\begin{align*}
    P_{i, \ell} := 
    \begin{cases}
        \{t \in P_i \mid \dist(x_t, c^*_i) \leq 2 \cdot \rmin\} & \text{ if } \ell = 1, \\
        \{t \in P_i \mid 2^{\ell - 1} \cdot \rmin \leq \dist(x_t, c^*_i) \leq 2^\ell \cdot \rmin\} & \text{ otherwise. }
    \end{cases}
\end{align*}
We note that $\bigcup_{\ell \in [L]} P_{i, \ell} = P_i$, since the diameter of $\cM$ is at most one.

Consider each $P_{i, \ell}$.  Referring back to \Cref{alg:meyerson-main}, recall that $Z_t$ is the outcome of the coin, determined by $\lambda_t$, that decides
if $x_t$ should be added to $S_t$.  
For every $j \in P_{i, \ell}$, we define $Z_{< j} := \sum_{\substack{j' \in P_{i, \ell} \\ j' < j}} Z_{j'}$, 
i.e., the number of elements in this partition that has been added until before time step $j$. 
Similarly, we define $\lambda_{< j} := \sum_{\substack{j' \in P_{i, \ell} \\ j' < j}} \lambda_{j'}$.

\paragraph{Proof of \eqref{eq:bicri-approx-expectation}.} Let us define additional notation:
\begin{align*}
\tZ_t =
\begin{cases}
Z_t &\text{ if } Z_{<t} < 3\tau, \\
0 &\text{ otherwise.}
\end{cases} & &
\tlambda_t =
\begin{cases}
\lambda_t &\text{ if } Z_{<t} < 3\tau, \\
0 &\text{ otherwise.}
\end{cases}
\end{align*}

Consider any $t \in P_{i, \ell}$.
Observe that, if $Z_{< t} \geq 3\tau$ (i.e.,  $\lambda_t \ne \tlambda_t$), then $|S_{t - 1} \cap \cB(c^*_i, 2^\ell \cdot \rmin)| \geq 3\tau$.
Let $t' < t$ be the last time step $\densealg$ is called before time step $t$. 
Due to the guarantee of $\countingalg$, we must have $|S_t \setminus S_{t'}| \leq \tau + 2 \cdot (\tau/2) = 2\tau$.
This implies that $Z_{< t'} \geq \tau$. By the guarantee of $\densealg$, there must be a point $f \in R_{\leq t - 1}$ such that $\dist(c^*_i, f) \leq O(2^\ell \cdot \rmin)$. This implies
\begin{align} \label{eq:distance-bound-after-sufficient-samples}
\dist(x_t, R_{\leq t - 1}) 
\leq \dist(x_t, c^*_i) + \dist(c^*_i, f) &\leq \dist(x_t, c^*_i) + O(2^\ell \cdot \rmin) \nonumber\\ &\leq O(\dist(x_t, c^*_i) + \bone[\ell = 1] \cdot \rmin).
\end{align}
Meanwhile, by definition of $\lambda_t$ and since the diameter of $\cM$ is at most one, we also have that
\begin{align}
\label{eq:naive-distance-bound}
\textstyle
\dist(x_t, R_{\leq t - 1})^p \leq \lambda_t \cdot \Paren{1 + \frac{\estopt}{\tau k L}}.  
\end{align}
Using \eqref{eq:distance-bound-after-sufficient-samples} when
$\lambda_t \neq \tlambda_t$ and \eqref{eq:naive-distance-bound} when $\lambda_t = \tlambda_t$, we have
\begin{align*}
\dist(x_t, R_{\leq t - 1})^p \leq O\Paren{\dist(x_t, c^*_i)^p + \bone[\ell = 1] \cdot\rmin^p + \tlambda_t \cdot \Paren{1 + \frac{\estopt}{\tau k L}} }.
\end{align*}
Taking expectation, we thus get
\begin{align*}
\E{\dist(x_t, R_{\leq t - 1})^p} \leq O\Paren{\dist(x_t, c^*_i)^p + \bone[\ell = 1] \cdot \rmin^p + \E{\tZ_t} \cdot \Paren{1 + \frac{\estopt}{\tau k L}}}.
\end{align*}
Summing over all $t \in [n]$, we get
\begin{align*}
\textstyle
\E{\sum_{t \in [n]} \dist(x_t, R_{\leq t - 1})^p} \leq O\Paren{\OPT_{X_{\leq n}}^{k, p} + n \cdot \rmin^p + \E{\sum_{t \in [n]}  \tZ_t} \cdot \Paren{1 + \frac{\estopt}{\tau k L}}}.
\end{align*}
Finally, by definition, $\sum_{t \in P_{i, \ell}}  \tZ_t \leq 3\tau$ regardless of the randomness of the algorithm. Summing over $i \in [k]$ and $\ell \in [L]$ gives $\sum_{t \in [n]}  \tZ_t \leq 3 \tau k L$. Plugging this into the above gives
\begin{align*}
\textstyle
\E{\sum_{t \in [n]} \dist(x_t, R_{\leq t - 1})^p} &\leq O\Paren{\OPT_{X_{\leq n}}^{k, p} + n \cdot \rmin^p + \tau k L + \estopt} \\ 
&\leq O\Paren{\OPT_{X_{\leq n}}^{k, p} + \estopt + \tau k L},
\end{align*}
where the last inequality is due to our choice of $\rmin$. This proves \eqref{eq:bicri-approx-expectation}.

\paragraph{Proof of \eqref{eq:size-s-expectation}.} First, recall that \eqref{eq:distance-bound-after-sufficient-samples} holds if $Z_{< t} \geq 3\tau$, i.e.,  $\tZ_t \ne Z_t$. In this case, we have 
\begin{align*}
\E{Z_t - \tZ_t} = \E{Z_t} = \lambda_t \overset{\eqref{eq:distance-bound-after-sufficient-samples}}{\leq} O\Paren{\frac{\tau k L \cdot \Paren{\dist(x_t, c^*_i)^p + \bone[\ell = 1] \cdot\rmin^p}}{\estopt}}.
\end{align*} 

Thus, the expected total number of points sampled from $P_{i, \ell}$ is
\begin{align*}
\textstyle
\E{\sum_{t \in P_{i, \ell}} Z_t} &= 
\textstyle
\E{\sum_{t \in P_{i, \ell}} \tZ_t} + \E{\sum_{t \in P_{i, \ell}} (Z_t - \tZ_t)} \\
&\leq 3\tau + O\Paren{\sum_{t \in P_{i, \ell}} \frac{\tau k L \cdot \Paren{\dist(x_t, c^*_i)^p + \bone[\ell = 1] \cdot\rmin^p}}{\estopt}}, 
\end{align*}
where in the first inequality we again use the fact that $\sum_{t \in P_{i, \ell}} \tZ_t \leq 3\tau$ by definition. 

Summing this up over all $i \in [k]$ and $\ell \in [L]$, we get 
\begin{align*}
\E{|S_{\leq n}|} \leq O\Paren{\tau k L + \frac{\tau k L \cdot \Paren{\OPT^{k, p}_{X_{\leq n}} + n \cdot \rmin^p}}{\estopt}} \leq O\Paren{\tau k L + \frac{\tau k L \cdot \OPT^{k, p}_{X_{\leq n}}}{\estopt}},
\end{align*}
where the second inequality is due to our choice of $\rmin$. Finally, the RHS is at most $\frac{\Psi \cdot \tau}{8}$ for sufficiently large constant $D$. Thus, \eqref{eq:size-s-expectation} holds as desired.
\end{proof}

The privacy and efficiency analysis of the sketch is deferred to Appendix~\ref{app:meyerson-privacy-efficiency}.
Finally, turning such a prefix-approximation into a coreset can be easily done via snapping each point $x_t$ to the closest center in the output so far. Again, we need use a continual counting algorithm to compute the weights in DP manner. This is a standard technique that has been applied before in the DP setting (e.g., \cite{GhaziKM20}). The full argument is deferred to Appendix~\ref{subsec:bicriteria-to-coreset}.
\section{From DP Semi-Coreset to DP Consistent Clustering}

As mentioned earlier, we can simply 
invoke the non-private online consistent clustering algorithm (\Cref{thm:consistent-additive}) on the private semi-coreset (\Cref{thm:main-semi-coreset}) to obtain the desired result.
\begin{corollary} \label{thm:dp-consistent-main}
There is an $(\eps, \delta)$-DP online $(k, p)$-clustering $\Paren{k \cdot \polylog{nd}}$-consistent algorithm whose output is an $(O(1), \beta)$-approximation with high probability, and the algorithm runs in time $\tO{ndk^2} + \Lambda^{O(1)}$ and uses space $\Lambda^{O(1)}$ where
\begin{itemize}
    \item $\beta = O\Paren{\frac{d \cdot k^3}{\eps} \cdot \polylog{\frac{n d}{\eps}} }$ and $\Lambda = \frac{kd \log n}{\eps}$ for pure-DP.
    \item $\beta = O\Paren{\frac{\sqrt{d} \cdot k^{2.5}}{\eps} \cdot \polylog{\frac{n d}{\eps \delta}} }$ and $\Lambda = \frac{kd \log\Paren{n/\delta} }{\eps}$ for approximate-DP.
\end{itemize}
\end{corollary}
\begin{proof}
We run the $(\eps, \delta)$-DP semi-coreset algorithm from \Cref{thm:main-semi-coreset} on input stream $X$ to get a multiset stream $\tX = (\tx_1, \dots, \tx_n)$. We then feed the points in $\tx_1, \dots, \tx_n$ to the online consistent $(k, p)$-clustering algorithm in \Cref{thm:consistent-additive}; the final output $C_t$ is the output from this algorithm after feeding it $\tx_t$. The consistency claim and time/space bounds  immediately follow. The privacy guarantee holds because of the post-processing property of DP, as we only post-process $(\tx_1, \dots, \tx_n)$.

Finally, since $C_t$ is an $(O(1), O(1))$-approximation of $\tX_{\leq t}$ and $\tX_{\leq t}$ is an $(O(1), O(1), \beta)$-semi-coreset of $X_{\leq t}$, \Cref{lem:coreset-approx} implies that $C_t$ is also an $(O(1), O(\beta))$-approximation of $X_{\leq t}$.
\end{proof}
\section{Conclusion and Open Questions}
\label{sec:conclusion}

In this work, we presented a generic reduction that transforms a sensitive data stream into a privatized semi-coreset stream. By allowing any non-private online clustering algorithm to be executed as a post-processing step, our framework yields the first DP online clustering algorithm that achieves a nearly optimal consistency guarantee, alongside a constant approximation ratio and a nearly optimal additive error in terms of the dimension. 

Despite these advances, several open questions and directions for future work remain:
\begin{itemize}[nosep]
\item \textbf{Practicality}:  While our private adaptation of Meyerson's sketch (\Cref{alg:meyerson-main}) is straightforward to implement and practical, the underlying DP dense ball algorithms it relies upon are theoretical. To our knowledge, there are currently no practical implementations of these algorithms, which poses a barrier to real-world deployment.
\item \textbf{Error Dependency on $k$:} While our additive error dependency on dimension $d$ is nearly optimal, the dependency on $k$ leaves room for improvement. In our private Meyerson's sketch (\Cref{alg:meyerson-main}), this sub-optimality is partly due to DP composition across all runs of the dense ball algorithm; it might be possible to improve on this by deriving an ``online'' version of DP dense ball. More generally, while efficient DP algorithms with tight additive errors in terms of $k$ are known in the offline setting~\citep{NguyenCX21,GhaziKKMS24}, none is known for the online setting. 
\item \textbf{Coreset with $(1 \pm o(1))$-Multiplicative Error:} Our online construction (and that of \cite{EMZ23}) only yields a semi-coreset, which naturally leaves the question of whether an online DP coreset construction (i.e., obtaining $\eta_1 = 0$ in \Cref{def:semi-coreset}) is possible. Another interesting question is whether one can reduce the multiplicative error $\kappa$ to $1 + o(1)$. We remark that the latter question is still open even in the offline setting, although there is some evidence that this might be hard~\citep{DP-coreset-hardness}.
\item \textbf{Adversarial Streaming Models:} While our privacy analysis holds for the adversarial setting (where the adversary sees outputs up to time $i - 1$ before selecting $x_i$), our current utility analysis assumes an oblivious data stream. 
There are several challenges. For instance, it has not yet been established whether the non-private Meyerson's sketch maintains its utility guarantee in the adversarial model. In fact, we are not aware of any analysis of the (non-private) consistent $k$-means/median clustering in this model. 
\item \textbf{Running Time Optimization:} Although our algorithm achieves $n^{1+o(1)}$ time complexity over the entire stream for $k, d \le n^{o(1)}$, the running time could be optimized further, e.g., by employing approximate nearest neighbor data structures. Furthermore, it also seems plausible that the $\Lambda$ term in our main theorem (\Cref{thm:main-semi-coreset})   can be optimized further. Currently, the polynomial dependency on $k, d$ through $\Lambda$ is due to the use of the 1-Cluster algorithm by \cite{GhaziKM20} as a blackbox (see Appendix~\ref{app:dense-ball}). It remains an interesting question whether this can be sped up further, e.g., to nearly linear time in $k$.
\end{itemize}  
\paragraph{AI Disclosure.} While all the main algorithms and proofs in this work were discovered by the authors, we used AI, including Gemini (Google) and Claude (Anthropic), to assist in preparing this paper. In particular, we used AI to flesh out proof sketches to first drafts of full proofs, and to help with drafting the introduction and conclusion. These drafts are heavily edited and verified by the authors, and we take full responsibility for all claims in this work.

\bibliographystyle{plainnat}
\bibliography{ref}

\appendix 
\section{Additional Background and Missing Proofs from \texorpdfstring{\Cref{sec:prelim}}{Section~\ref{sec:prelim}}} \label{app:prelim}

When we state the running time and space complexity of an algorithm, we always assume that each coordinate of $x_i$ fits in a single word of memory and arithmetic operations on them take constant time.  This, in particular, means that reading each $x_i$ takes 
$O(d)$ time and storing each $x_i$ uses $O(d)$ space. 


\paragraph{Approximate Triangle Inequality.}
The following ``approximate triangle inequality'' for the $p$th power will be useful in our proofs.

\begin{lemma}[Approximate Triangle Inequality, \citep{MakarychevMR19}] \label{lem:approx_triangle}
For any $p \ge 1$ and $\xi \in (0,1]$, there exists a constant $\vartheta_{p, \xi} > 0$ that depends only on $p$ and $\xi$ such that for any $a, b, c$ in a metric space, 
\begin{align*}
    \dist(a, b)^p \leq (1 + \xi) \cdot \dist(a, c)^p + \vartheta_{p, \xi} \cdot \dist(c, b)^p.
\end{align*}
\end{lemma}

\subsection{Properties of Semi-Coreset}

\subsubsection{Monge's Optimal Transport and Semi-Coreset}

We use a generalization of the optimal transport~\citep{monge1781memoire}, as defined in \citep{ChangGKM21}.

\begin{definition}[Generalized Transport Cost, \citep{ChangGKM21}]
Let $X$ and $\tX$ be weighted point sets over a metric space $U$. For a mapping $\Psi: U \to U$, the \emph{generalized $p$-transport cost} is defined as:
\begin{align} \label{eq:gen-transport}
    \mt(\Psi, X, \tX) := \sum_{x \in U} w_X(x) \cdot \dist(x, \Psi(x))^p + \sum_{y \in U} |w_X(\Psi^{-1}(y)) - w_{\tX}(y)|,
\end{align}
where we use the notation $w_X(x)$ to denote the weight of point $x$ in $X$ and $w_X(S) = \sum_{x \in S} w_X(x)$.
\end{definition}

In other words, instead of allowing just the distance error (first term), we also allow weight mismatch (second term). This make it convenient for DP, which naturally produces noisy weights.

In \cite{ChangGKM21}, it was shown that if two sets have ``small'' optimal transport cost, then one is a good coreset of the other. However, the ``small'' cost requirement in \cite{ChangGKM21} is restrictive: It requires $\mt(\Psi, X, \tX) \leq \alpha \cdot \OPT^{k,p}_X + \beta$ for $\alpha < 1$.  On the other hand, our proof will use $\alpha > 1$. Nevertheless, we show that the statement is still true for $\alpha > 1$, albeit with semi-coreset instead of coreset. This is formalized in the lemma below, which will be handy in our proofs.

\begin{lemma}[Transport $\implies$ Semi-Coreset] \label{lem:transport_coreset}
Assume that all points in $U$ have a pairwise distance bounded by 1 (i.e., diameter $\le 1$). Let $X$ and $\tX$ be weighted point sets over $U$. If there exists a mapping $\Psi: X \to U$ such that $\mt(\Psi, X, \tX) \leq \alpha \cdot \OPT^{k,p}_X + \beta$, then $\tX$ is a $(\kappa, \eta_1, \eta_2)$-semi-coreset for $X$ with parameters $\kappa = O(1 + \alpha)$, $\eta_1 = O(\alpha)$, and $\eta_2 = O(\beta)$.
\end{lemma}

\begin{proof}
Consider any set $C \in \binom{U}{k}$ of $k$ centers. For any point $y \in U$, since the diameter of the metric space is bounded by 1, we have $\dist(y, C)^p \leq 1$. By applying the approximate triangle inequality (Lemma \ref{lem:approx_triangle}) with $\xi = 1$, we can bound the cost of $C$ on the estimated set $\tX$:
\begin{align*}
    &\cost^p_{\tX}(C)\\
    &= \sum_{y \in U} w_{\tX}(y) \cdot \dist(y, C)^p \\
    &\leq \sum_{y \in U} w_X(\Psi^{-1}(y)) \cdot \dist(y, C)^p + \sum_{y \in U} |w_{\tX}(y) - w_X(\Psi^{-1}(y))| \cdot \dist(y, C)^p \\
    &\leq \sum_{x \in U} w_X(x) \cdot \dist(\Psi(x), C)^p + \sum_{y \in U} |w_{\tX}(y) - w_X(\Psi^{-1}(y))| \\
    &\leq \sum_{x \in U} w_X(x) \cdot \left( 2 \cdot \dist(x, C)^p + \vartheta_{p,1} \cdot \dist(x, \Psi(x))^p \right) + \sum_{y \in U} |w_{\tX}(y) - w_X(\Psi^{-1}(y))| \\
    &\leq 2 \cost^p_X(C) + \max(\vartheta_{p,1}, 1) \cdot \mt(\Psi, X, \tX).
\end{align*}
Using the bounds $\mt(\Psi, X, \tX) \leq \alpha \cdot \OPT^{k,p}_X + \beta$ and  $\OPT^{k,p}_X \leq \cost^p_X(C)$, we obtain the upper bound:
\begin{align*}
    \cost^p_{\tX}(C) &\leq 2 \cost^p_X(C) + \max(\vartheta_{p,1}, 1) \left( \alpha \cdot \OPT^{k,p}_X + \beta \right) \\
    &\leq \left(2 + \max(\vartheta_{p,1}, 1) \cdot \alpha \right) \cost^p_X(C) + \max(\vartheta_{p,1}, 1) \beta.
\end{align*}
A symmetric argument upper bounds the true cost:
\begin{align*}
    \cost^p_X(C) \leq 2 \cost^p_{\tX}(C) + \max(\vartheta_{p,1}, 1) \left( \alpha \cdot \OPT^{k,p}_X + \beta \right).
\end{align*}
Rearranging this to isolate $\cost^p_{\tX}(C)$ yields the lower bound:
\begin{align*}
    \cost^p_{\tX}(C) \geq \frac{1}{2} \cost^p_X(C) - \frac{\max(\vartheta_{p,1}, 1) \cdot \alpha}{2} \OPT^{k,p}_X - \frac{\max(\vartheta_{p,1}, 1)}{2} \cdot \beta.
\end{align*}
By defining $\kappa = \max(2 + \max(\vartheta_{p,1}, 1) \cdot \alpha, 2) = O(1 + \alpha)$, $\eta_1 = \frac{\max(\vartheta_{p,1}, 1) \cdot \alpha}{2} = O(\alpha)$, and $\eta_2 = \max(\vartheta_{p,1}, 1) \cdot \beta = O(\beta)$, both the upper and lower bounds of the semi-coreset definition are satisfied.
\end{proof}

\subsubsection{Proof of \texorpdfstring{\Cref{lem:coreset-approx}}{Lemma~\ref{lem:coreset-approx}}}

\begin{proof}[Proof of \Cref{lem:coreset-approx}]
Let $C^*_{X}$ denote the optimal set of $k$ centers for $X$ such that $\cost^p_{X}(C^*_{X}) = \OPT^{k,p}_{X}$. By the definition of a $(\kappa, \eta_1, \eta_2)$-semi-coreset applied to $C^*_{X}$, we have the upper bound
\begin{align*}
\cost^p_{\tX}(C^*_{X}) \leq \kappa \cdot \cost^p_{X}(C^*_{X}) + \eta_2 
= \kappa \cdot \OPT^{k,p}_{X} + \eta_2.
\end{align*}

Since $\OPT^{k,p}_{\tX} \leq \cost^p_{\tX}(C)$ for any $C \in \binom{U}{k}$, it follows that $\OPT^{k,p}_{\tX} \leq \cost^p_{\tX}(C^*_{X})$.
By assumption, the set $C$ of centers is an $(\alpha, \beta)$-approximation for $\tX$. Therefore,
\begin{align}
\cost^p_{\tX}(C) &\leq \alpha \cdot \OPT^{k,p}_{\tX} + \beta \nonumber \\
&\leq \alpha \cdot \cost^p_{\tX}(C^*_{X}) + \beta \nonumber \\
&\leq \alpha \cdot \left( \kappa \cdot \OPT^{k,p}_{X} + \eta_2 \right) + \beta \nonumber \\
&= \alpha \kappa \cdot \OPT^{k,p}_{X} + \alpha \eta_2 + \beta. \label{eq:approx_bound}
\end{align}

Now, we apply the lower bound property of the $(\kappa, \eta_1, \eta_2)$-semi-coreset to the centers $C$:
\begin{align}
\frac{1}{\kappa} \cdot \cost^p_{X}(C) - \eta_1 \cdot \OPT^{k, p}_{X} - \eta_2 \leq \cost^p_{\tX}(C). \label{eq:coreset_lower}
\end{align}

Substituting the upper bound from \eqref{eq:approx_bound} into \eqref{eq:coreset_lower}, we obtain:
\begin{align*}
\frac{1}{\kappa} \cdot \cost^p_{X}(C) - \eta_1 \cdot \OPT^{k, p}_{X} - \eta_2 \leq \alpha \kappa \cdot \OPT^{k,p}_{X} + \alpha \eta_2 + \beta.
\end{align*}

Rearranging to isolate $\cost^p_{X}(C)$ yields:
\begin{align*}
\frac{1}{\kappa} \cdot \cost^p_{X}(C) &\leq (\alpha \kappa + \eta_1) \cdot \OPT^{k,p}_{X} + (\alpha + 1) \eta_2 + \beta \\
\cost^p_{X}(C) &\leq (\alpha \kappa^2 + \eta_1 \kappa) \cdot \OPT^{k,p}_{X} + \kappa(\alpha + 1)\eta_2 + \kappa \beta.
\end{align*}
This completes the proof.
\end{proof}

\subsection{Consistent \texorpdfstring{$k$}{k}-Clustering via Grid Snapping: Proof of \texorpdfstring{\Cref{thm:consistent-additive}}{Theorem~\ref{thm:consistent-additive}}}
\label{subsec:consistent-snapping}

As mentioned in \Cref{subsec:prelim-consistent}, the original consistency bound in \cite{FichtenbergerLN21} depends on the \emph{aspect ratio} of the metric space, which is defined as the ratio between the largest distance and the smallest non-zero distance. Their result can be stated as follows:

\begin{theorem}[\cite{FichtenbergerLN21}] \label{thm:fichtenberger-consistent}
Let $\cM = (U, d)$ be any metric space with aspect ratio $\Delta$. There is an online $(k, p)$-clustering $\Paren{k \cdot \polylog{n\Delta}}$-consistent algorithm for metric space $\cM$ whose output is an $(O(1), 0)$-approximation with high probability.
Moreover, the algorithm has time and space complexity at most polynomial in the number of distinct points in the input stream and the time needed to compute distances between points.
\end{theorem}

To bypass the aspect ratio assumption in consistent clustering, we transform the input stream by snapping points to a fine grid. We show that this transformation controls the aspect ratio of the active metric space while incurring only a constant additive error, which we can formalize via the generalized transport cost.

\begin{definition}[Grid-Snapped Stream] \label{def:grid-snapped}
Let $X = (x_1, \dots, x_n)$ be a stream of points in $\R^d$ with a metric diameter at most $1$. Define the grid $G = (\frac{1}{nd}\Z)^d$. Let the snapped stream be $\hat{X} = (\hat{x}_1, \dots, \hat{x}_n)$, where each $\hat{x}_t$ is obtained by rounding each coordinate of $x_t$ to a multiple of $\frac{1}{nd}$ in such a way that its absolute value does not increase.
\end{definition}

\begin{lemma}[Snapped Stream as a Semi-Coreset] \label{lem:snapped-coreset}
For any $t \in [n]$, the snapped multiset $\hat{X}_{\le t}$ is a $(\kappa, 0, \eta_2)$-semi-coreset for $X_{\le t}$ for the $(k, p)$-clustering objective, where $\kappa = O(1)$ and $\eta_2 = O(1)$.
\end{lemma}

\begin{proof}
Consider the natural mapping $\Psi: X_{\le t} \to \hat{X}_{\le t}$ given by $\Psi(x_i) = \hat{x}_i$. Because the point weights match perfectly, the weight difference term in the generalized transport cost \eqref{eq:gen-transport} is 0. 

For the distance term in \eqref{eq:gen-transport}, rounding each coordinate to displaces a point by at most $\frac{1}{nd}$ along any dimension. The $\ell_2$ displacement is therefore bounded by:
\begin{align*}
    \dist(x_i, \hat{x}_i) \leq \sqrt{d \left(\frac{1}{nd}\right)^2} = \frac{\sqrt{d}}{nd} \leq \frac{1}{n\sqrt{d}} \leq \frac{1}{n}.
\end{align*}
For the $(k, p)$-clustering objective, the generalized $p$-transport cost is:
\begin{align*}
    \mt(\Psi, X_{\le t}, \hat{X}_{\le t}) = \sum_{i=1}^t \dist(x_i, \hat{x}_i)^p \leq t \cdot \frac{1}{n} \leq 1.
\end{align*}
By applying Lemma \ref{lem:transport_coreset}, $\hat{X}_{\le t}$ is an $(O(1), 0, O(1))$-semi-coreset for $X_{\le t}$.
\end{proof}

We are now ready to prove \Cref{thm:consistent-additive}.

\begin{proof}[Proof of \Cref{thm:consistent-additive}]
Instead of clustering the original stream $X$, we maintain and cluster the grid-snapped stream $\hat{X}$ as defined above (\Cref{def:grid-snapped}). Since all points in $\hat{X}$ are restricted to the grid $G = (\frac{1}{nd}\Z)^d$ and the overall diameter is at most 1, the maximum distance between any two points in $\hat{X}$ is 1. Conversely, the minimum non-zero distance between any two distinct grid points is at least $\frac{1}{nd}$. Therefore, the aspect ratio $\Delta$ of the active points in $\hat{X}$ is strictly bounded by $nd$.

We run the consistent $k$-clustering algorithm from Theorem \ref{thm:fichtenberger-consistent} over the snapped stream $\hat{X}$. Because $\Delta \le nd$, this guarantees a total consistency cost of $k \cdot \polylog{n \Delta} \le k \cdot \polylog{n d}$. 
    
For the utility guarantee, Theorem \ref{thm:fichtenberger-consistent} ensures that with high probability, for all $t \in [n]$, $C_t$ is an $(O(1), 0)$-approximation on the snapped points $\hX_{\leq t}$.
By Lemma \ref{lem:snapped-coreset}, $\hat{X}_{\leq t}$ is a $(O(1), 0, O(1))$-semi-coreset for $X_{\leq t}$. Applying \Cref{lem:coreset-approx}, we can thus conclude that $C_t$ is an $(O(1), O(1))$-approximation for the original set $X_{\leq t}$, which concludes our proof.
\end{proof}

\subsection{Differential Privacy}

As mentioned in \Cref{sec:prelim-dp}, our privacy proof holds in the stronger \emph{adaptive} model as well. Here we follow terminologies from \cite{Jain23}, although similar models were considered even before (e.g. \cite{ThakurtaS13}). The adaptive setting can be defined using a game between the adversary and the algorithm as follows.

\begin{definition}[Adaptive continual release game]\label{def:game}
Let $\mathcal{B}$ be an adversary and $b\in\{0,1\}$. The game $\Expt^b(\mathcal{B},\mathcal{M})$
proceeds as follows.
\begin{enumerate}
  \item For $t=1,\dots,n$, given the transcript $\pi_{t-1} = (o_1, \dots, o_{t-1})$, the adversary declares
        $\mathrm{type}_t\in\{\textsf{regular},\textsf{challenge}\}$, where \textsf{challenge} is
        declared at most once over the whole game.
        \begin{itemize}
          \item If $\mathrm{type}_t=\textsf{regular}$, $\mathcal{B}$ outputs a single
            $x_t$, which is sent to $\mathcal{M}$ in both worlds.
          \item If $\mathrm{type}_t=\textsf{challenge}$, set $i^\ast\leftarrow t$;
            $\mathcal{B}$ outputs a pair $(x_t^0,x_t^1)$ with
            $\perp\in\{x_t^0,x_t^1\}$, and $x_t^b$ is sent to $\mathcal{M}$.
        \end{itemize}
  \item $\mathcal{M}$ releases $o_t$, which is appended to the transcript.
  \item The game outputs $\pi_n$.
\end{enumerate}
\end{definition}

\begin{definition}[DP under adaptive inputs]\label{def:adaptivedp}
$\mathcal{M}$ is $(\eps,\delta)$-DP in the adaptive continual release model if for every
adversary $\mathcal{B}$ and every measurable set $T$ of transcripts,
\[
  \Pr\bigl[\Expt^0(\mathcal{B},\mathcal{M})\in T\bigr]
  \;\le\; e^{\eps}\,\Pr\bigl[\Expt^1(\mathcal{B},\mathcal{M})\in T\bigr]+\delta .
\]
\end{definition}

It is clear that \Cref{def:adaptivedp} is no weaker than that in \Cref{def:nonadaptivedp}: For a fixed pair neighboring streams $X = (x_1, \dots, x_n), X' = (x'_1, \dots, x'_n)$, the adaptive adversary can always simply set $x_t^0 = x_t$ and $x_t^1 = x'_t$ for all $t \in [n]$ and recovers the non-adaptive guarantee.

Henceforth, for online algorithms, we will always consider privacy in the stronger adaptive continual release model (\Cref{def:adaptivedp}). For brevity, we will sometimes not state this explicitly.

\subsubsection{Tools from Differential Privacy}

We next state two lemmas which will be helpful in the privacy proofs. First is the so-called ``post-processing'' property of DP, specialized to the adaptive continual release model below.

\begin{lemma}[Online post-processing]\label{lem:postproc}
Let $\mathcal{M}$ be $(\eps,\delta)$-DP in the adaptive continual release model, and let
$\mathcal{M}'$ be the online mechanism that runs $\mathcal{M}$ internally and, at step $t$,
releases $g_t(\pi_t;\sigma)$, where $\pi_t = (o_1, \dots, o_t)$ is $\mathcal{M}$'s transcript prefix, $g_t$ is a
fixed function, and $\sigma$ is randomness independent of the input. Then $\mathcal{M}'$ is
$(\eps,\delta)$-DP in the adaptive continual release model.
\end{lemma}

We also need composition properties of DP. Since we employ multiple online DP subroutines which are interleaved, we require a ``concurrent'' version of DP composition, as stated below. 


\begin{lemma}[Concurrent composition; \cite{HSV24}]\label{lem:concurrent}
Let $\mathcal{M}_1,\dots,\mathcal{M}_m$ be continual mechanisms run concurrently, where each $\mathcal{M}_i$ is $(\eps_i,\delta_i)$-DP in the adaptive model, with respect to a
neighbor relation $\sim_i$ on its input streams.
Consider an adversary that interacts with all $m$ mechanisms concurrently, interleaving its interactions arbitrarily and choosing each input to each mechanism
adaptively as a function of all outputs received so far from all
mechanisms.
 Then the composition is
$\bigl(\sum_i \eps_i,\ \sum_i \delta_i\bigr)$-DP with respect to any pair of adaptive strategies whose induced input streams satisfy $\sim_i$ for every $i$. 
\end{lemma}



We also need a ``parallel'' composition theorem for continual mechanisms, in which the neighbor relation on the composed input is restricted so that the two worlds differ only in the inputs to a single mechanism. Under this restriction, the privacy loss does not grow with $m$. We state this for the pure-DP case, which is all we need. 
\begin{lemma}[Parallel composition; \cite{HSV24}]\label{lem:parallel}
Consider the setting of \Cref{lem:concurrent}, with the neighbor
relation restricted as follows: the input streams induced by the two strategies are identical for all but one mechanism
$\mathcal{M}_{i^*}$, and neighboring for $\mathcal{M}_{i^*}$. If
$\eps_1 = \cdots = \eps_m = \eps'$ and $\delta_1 = \cdots = \delta_m = 0$, then the composed mechanism is $\eps'$-DP with respect to this relation.
\end{lemma}

\subsubsection{Differentially Private Continual Counting}

For continual counting, two input streams $Y = (y_1, \dots, y_n) \in \{0, 1\}^n$ and $Y' = (y'_1, \dots, y'_n) \in \{0, 1\}^n$ are neighbors iff they differ on a single coordinate.

Throughout, we will assume that the estimates $e_t$ of $\sum_{i \in [t]} y_i$ are monotone, i.e., $e_1 \leq \cdots \leq e_n$, and that they are integers.  (If this does not hold, it can be easily achieved by taking $\lceil \max_{i \in [t]} e_i \rceil$ instead of $e_t$.)  The seminal work of~\cite{DworkNPR10, ChanSS11} gives a DP algorithm that is $O(\polylog{n})$-accurate with high probability. This is stated more formally below.

\begin{theorem}[\cite{DworkNPR10,ChanSS11}] \label{thm:cont-counting}
There exists an $\eps$-DP\footnote{Here the DP definition as in \Cref{def:adaptivedp} with 0 being treated as $\perp$.} continual counting algorithm that, for any $\gamma > 0$, is $O\Paren{\frac{1}{\eps} \cdot \log^{1.5} n \cdot \log(1/\gamma)}$-accurate with probability $1 - \gamma$. The total running time of the algorithm is $O(n \log n)$ and its space complexity is $O(\log n)$.
\end{theorem}
\subsection{Dense Ball Algorithm} \label{app:dense-ball}

For the dense ball problem (\Cref{def:dense-ball}), we consider two input datasets $S, S'$ neighbors iff $|S \Delta S'| \leq 1$, where $\Delta$ denote the symmetric difference between the two sets.

Our definition of the dense ball problem is closely related to the 1-Cluster and the densest ball problems previously studied in DP literature~\citep{NissimSV16,NissimS18,GhaziKM20}. By repeatedly applying a known algorithm for this problem, it is simple to derive the following result for the dense ball problem. 

\begin{theorem} \label{thm:dense-ball-main}
Let $N \in \N$ and suppose that the input $S$ has size at most $N$. Then, there exists an $(\eps, \delta)$-DP algorithm with time and space complexity $(N d \cdot \log(1/r_{\min}))^{O(1)}$ and size blow-up of $O(1)$, which solves the $(r_{\min}, \tau)$-dense ball with probability $1 - \gamma$ where
\begin{itemize}
\item $\tau = O\Paren{\Paren{\frac{d \cdot N}{\eps}}^{1/2} \cdot \polylog{\frac{N d}{\eps \rmin}}}$ for $\delta = 0$, and,
\item $\tau = O\Paren{\Paren{\frac{d \cdot N}{\eps^2}}^{1/3} \cdot \polylog{\frac{N d}{\eps \delta \rmin}}}$ for $\delta > 0$.
\end{itemize}
\end{theorem} 

Our construction builds upon the DP 1-Cluster algorithms of \cite{GhaziKM20}. We formally define the problem and recall their result below.

\begin{definition}[1-Cluster, \citep{NissimSV16}] \label{def:1-cluster}
Given a multiset $S$ of size at most $N$ in a metric space, a target number of points $\tau$, a minimum radius, and an approximation ratio $\alpha \ge 1$, the \emph{1-Cluster problem} asks to output a center $c$ and a radius $r$ such that the ball $\cB(c, r)$ contains at least $\tau - t$ points from $S$ (where $t$ is the additive error), and $r \le \alpha \cdot \max\{\rmin, r_{\opt}\}$, where $r_{\opt}$ is the minimum radius of a ball containing at least $\tau$ points in $S$.
\end{definition}

\begin{theorem}[DP 1-Cluster, \citep{GhaziKM20}] \label{thm:ghazi-1-cluster}
There exist $(\eps, \delta)$-DP algorithms for the 1-Cluster problem in the Euclidean space with approximation ratio $\alpha \in (1, 2)$ that succeed with probability $1 - \beta$ and have an additive error bounded by:
\begin{itemize}
    \item $t = O\Paren{\frac{d}{\eps} \cdot \polylog{\frac{N d}{\eps r_{\min} \beta}}}$ for pure-DP ($\delta = 0$).
    \item $t = O\Paren{\frac{\sqrt{d}}{\eps} \cdot \polylog{\frac{N d}{\eps \delta r_{\min} \beta}}}$ for approximate-DP ($\delta > 0$).
\end{itemize}
The time and space complexities are bounded by $(Nd \cdot \log(1/r_{\min}))^{O(1)}$.
\end{theorem}

We quickly note that, in the case of approximate-DP, the algorithm from \citep{GhaziKM20} in fact has an even better dependency on $r_{\min}$. However, we choose to state this coarser error bound since it does not effect our final guarantee anyway as we only choose $r_{\min} = \frac{1}{n}$ in the subsequent step.

We formalize our procedure in Algorithm \ref{alg:dense-ball-reduction}. The algorithm iteratively extracts dense balls from the dataset by repeatedly invoking the DP 1-Cluster subroutine.

\begin{algorithm}[H]
\caption{Dense Ball via Iterative 1-Cluster}
\label{alg:dense-ball-reduction}
\begin{algorithmic}[1]
\REQUIRE Multiset $S$ of size at most $N$, minimum target radius $r_{\min}$, target points $\tau$, privacy parameters $(\eps, \delta)$, failure probability $\gamma$.
\STATE $O \gets \emptyset$
\STATE $S_1 \gets S$
\STATE $K \gets \lceil 2N/\tau \rceil$ \hfill \algcomment{Number of iterations}
\STATE Set per-run privacy parameters $(\eps', \delta')$ based on composition over $K$ total invocations.
\FOR{$i = 1, \dots, K$}
    \STATE $(c_i, r_i) \gets$ Result of running the 1-Cluster algorithm of (\Cref{thm:ghazi-1-cluster}) on $S_i$ target $\tau$, approximation ratio $\alpha=O(1)$, failure probability $\beta = \gamma/K$, and privacy $(\eps', \delta')$. 
    \STATE $O \gets O \cup \{c_i\}$
    \STATE $S_{i+1} \gets S_i \setminus \cB(c_i, r_i)$
\ENDFOR
\RETURN $O$
\end{algorithmic}
\end{algorithm}

\begin{proof}[Proof of Theorem \ref{thm:dense-ball-main}]
We analyze \Cref{alg:dense-ball-reduction} by establishing its privacy, accuracy (size bound and coverage), and matching the required parameter bounds for $\tau$.

\paragraph{Privacy.} 
\Cref{alg:dense-ball-reduction} invokes the base 1-Cluster algorithm a total of $K$ times. We distribute the total privacy budget $(\eps, \delta)$ across all $K = O(N/\tau)$ invocations using composition theorems:
\begin{itemize}
    \item Pure-DP: By basic composition, setting $\eps' = \eps / K$ ensures that \Cref{alg:dense-ball-reduction} is $\eps$-DP.
    \item Approximate-DP: By advanced composition, to achieve overall $(\eps, \delta)$-DP, it is sufficient to set $\delta' = \Theta\Paren{\frac{\delta}{K}}$ and $\eps' = \Theta\Paren{\frac{\eps}{\sqrt{K \ln(1/\delta)}}}$.
\end{itemize}

\paragraph{Size Bound.} The algorithm runs for exactly $K$ steps, outputting exactly one center per step. Thus, the output set $O$ has size $|O| = K \leq 4N/\tau$. This satisfies the size bound property with blow-up factor $\Upsilon = 4$.

\paragraph{Coverage.} 
Let $t$ denote the additive error of the 1-Cluster algorithm. As specified in more detail below, we will parameterize the noise such that $t \le \tau/2$. 

Since each run of the 1-Cluster algorithm is successful with probability $1 - \beta$, all the runs succeed with probability at least $1 - \beta \cdot K = 1 - \gamma$. Conditioned on this, we will show that the coverage constraint is satisfied. Suppose there exist $u \in \R^d, r \ge r_{\min}$ such that $|S \cap \cB(u, r)| \ge \tau$. We must show that $\dist(u, O) \le \digamma \cdot r$.

First, let us assume for contradiction that no removed ball $\cB(c_i, r_i)$ intersects $\cB(u, r)$. If this were true, none of the points in $S \cap \cB(u, r)$ would ever be removed. Thus, at every step $i$, the remaining set $S_i$ would contain all $\ge \tau$ points of $S \cap \cB(u, r)$. Because the algorithm succeeds at every step and $t \le \tau/2$, each of the $K$ steps removes at least $\tau/2$ points from the dataset. Across $K$ steps, the algorithm would remove $K \cdot (\tau/2) = \lceil 2N/\tau \rceil \cdot (\tau/2) \ge N$ points. Since the total dataset has size at most $N$, this implies all points are eventually removed, contradicting the assumption that the points in $\cB(u, r)$ remain untouched.

Therefore, there must be at least one step where the removed ball $\cB(c_i, r_i)$ intersects $\cB(u, r)$. Let $j$ be the \emph{first} such step. At the start of step $j$, none of the points in $\cB(u, r)$ have been removed, meaning $|S_j \cap \cB(u, r)| \ge \tau$. This implies that the optimal radius $r_{\opt}$ for capturing $\tau$ points in $S_j$ is bounded by $r$ (i.e., $r_{\opt} \le r$). Since we ran 1-Cluster with approximation ratio $\alpha = O(1)$, the output radius guarantees $r_j \le \alpha \cdot \max\{\rmin, r_{\opt}\} \le \alpha \cdot r$. Since $\cB(c_j, r_j)$ intersects $\cB(u, r)$, let $x$ be a point in their intersection. By the triangle inequality, we have:
    \begin{align*}
        \dist(u, c_j) &\le \dist(u, x) + \dist(x, c_j) \le r + r_j \le 3 r.
    \end{align*}
    Since $c_j \in O$, this proves the required coverage.

\paragraph{Parameter Settings.} 
To complete the proof, we select the parameters to ensure $t \le \tau/2$.

\paragraph{For Pure-DP:} 
    By Theorem \ref{thm:ghazi-1-cluster}, the additive error of 1-Cluster is $t = O\left( \frac{d}{\eps'} \cdot \polylog{\frac{N d}{\eps' r_{\min} \gamma}} \right)$. Substituting $\eps' = \eps / K$, we get:
    \begin{align*}
        t = O\left( \frac{d N}{\eps \tau} \cdot \polylog{\frac{N d}{\eps r_{\min} \gamma}}\right)
    \end{align*}
    Enforcing $t \le \tau/2$ yields the requirement:
    \begin{align*}
        \tau^2 \ge \Theta\left( \frac{d N}{\eps} \cdot \polylog{\frac{N d}{\eps r_{\min} \gamma}} \right) \iff \tau \geq \Theta\left( \left( \frac{d N}{\eps} \right)^{1/2} \cdot \polylog{\frac{N d}{\eps r_{\min} \gamma}}\right).
    \end{align*}

\paragraph{For Approximate-DP:} 
    By Theorem \ref{thm:ghazi-1-cluster}, the additive error is $t = O\left( \frac{\sqrt{d}}{\eps'} \cdot \polylog{\frac{N d}{\eps \delta r_{\min} \gamma}} \right)$. Substituting our advanced composition parameter $\eps' = \Theta\left(\frac{\eps}{\sqrt{K \log(1/\delta)}}\right) = \Theta\left(\frac{\eps \sqrt{\tau}}{\sqrt{N \log(1/\delta)}}\right)$, we get:
    \begin{align*}
        t = O\left( \frac{\sqrt{d N}}{\eps \sqrt{\tau}} \cdot \polylog{\frac{N d}{\eps \delta r_{\min} \gamma}} \right)
    \end{align*}
    Enforcing $t \le \tau/2$ yields the requirement:
    \begin{align*}
        \tau^{3/2} \ge \Theta\left( \frac{\sqrt{d N}}{\eps} \polylog{\frac{N d}{\eps \delta r_{\min} \gamma}} \right) \iff \tau \ge \Theta\left( \left( \frac{d N}{\eps^2} \right)^{1/3} \polylog{\frac{N d}{\eps \delta r_{\min} \gamma}} \right).
    \end{align*}

\paragraph{Time and Space Complexity.} 
The 1-Cluster subroutine is called $K \leq N$ times. Each run takes $(N d \log(1/r_{\min}))^{O(1)}$ time and space. Thus, the total time and space complexity is bounded by $(N d \log(1/r_{\min}))^{O(1)}$.
\end{proof}
\section{Missing Proofs from \texorpdfstring{\Cref{sec:dp-meyerson-main}}{Section~\ref{sec:dp-meyerson-main}}}

\subsection{Private Meyerson's Sketch: Privacy and Efficiency} \label{app:meyerson-privacy-efficiency}

In this section, we state generic bounds on the privacy and efficiency guarantees of our DP Meyerson's sketch (\Cref{alg:meyerson-main}).

\begin{lemma}[Privacy] \label{lem:dpmyersen-privacy}
If $\densealg$ is $(\eps_d, \delta_d)$-DP and $\countingalg$ is $(\eps_c, \delta_c)$-DP, then \Cref{alg:meyerson-main} is $\Paren{2 \cdot \Psi \cdot \eps_d + \eps_c, 2 \cdot e^{2\eps_d} \cdot \Psi \cdot \delta_d + \delta_c}$-DP.
\end{lemma}

\begin{proof}
Fix an adversary $\mathcal{B}$. Recall that we use $i^\ast$ to denote the challenge index. 

First, it will be more convenient to consider a modified algorithm $\widehat{\mathcal{A}}$, which is exactly the same as \Cref{alg:meyerson-main}, except it releases $(e_t, R_t)$ at each step (instead of just $R_t$). By
\Cref{lem:postproc}, it suffice to prove that $\widehat{\mathcal{A}}$ is $\Paren{2 \cdot \Psi \cdot \eps_d + \eps_c, 2 \cdot e^{2\eps_d} \cdot \Psi \cdot \delta_d + \delta_c}$-DP.

Moreover, in \Cref{alg:meyerson-main}, we write $Z_t \sim \Bern(\lambda_t)$. For our purpose, it will be more transparent to instead think of this as setting $Z_t=\mathbf{1}[U_t\le\lambda_t]$ where $U=(U_1,\dots,U_n)$ is drawn i.i.d.\ uniform on $[0,1]$. Let $\widehat{\mathcal{A}}_u$ denote $\widehat{\mathcal{A}}$
with $U$ fixed to $u$. 
By the convexity of DP (see,  e.g.,  \cite{DesfontainesP20}), it suffices to prove the bound for every fixed $u$.

Fix $u = (u_1, \ldots, u_n)$ and a transcript prefix $\pi_{t-1} = ((e_1, R_1), \dots, (e_{t-1}, R_{t-1}))$. Given $\pi_{t-1}$, the set $R_{\le t-1}$ and the
counter $j$ are determined, hence so is the map
$\lambda(\cdot)=\min\{1,\frac{\tau kL}{\estopt}\dist(\cdot,R_{\le t-1})^p\}$. The adversary is
also a function of $\pi_{t-1}$, so it submits the same $x_t$ in both worlds for $t\neq i^\ast$.
Therefore, the bit fed to $\countingalg$ at step $t$ is the deterministic function
$z_t=\mathbf{1}[u_t\le\lambda(x_t)]$ of the submitted point, and we have the following. (We use the superscript $b \in \{0, 1\}$ to denote the execution on $(x_1^b, \dots, x_n^b)$.)
\begin{enumerate}\itemsep2pt
  \item[(a)] \emph{The counting streams are neighbors.} For $t\neq i^\ast$, we have
  $x_t^0=x_t^1$ and hence $z_t^0=z_t^1$. So the two streams differ in at most the coordinate $i^\ast$.
  \item[(b)] \emph{The dense-ball inputs are $2$-neighbors.}  Assume w.l.o.g. that $x_{i^\ast}^1=\perp$. 
  Both executions insert points in stream order,
  skipping an arrival once $|S|=\Psi\tau$. We thus have $S_t^0\setminus S_t^1\subseteq\{x_{i^\ast}^0\}$ and $|S_t^1\setminus S_t^0|\le 1$,
    Hence, $|S_t^0\,\Delta\,S_t^1|\le 2$.
\end{enumerate}

Note that $\widehat{\mathcal{A}}_u$ is the concurrent composition of $\countingalg$, run on
a bit stream generated adaptively from the transcript, with at most $\Psi$ invocations of
$\densealg$, each on an input determined by the transcript, $u$ and the submitted points.
By (a), the first component is $(\eps_c,\delta_c)$-DP; by (b) and
group privacy for groups of size $2$, each invocation of $\densealg$ is
$(2\eps_d,2e^{2\eps_d}\delta_d)$-DP. \Cref{lem:concurrent} then gives the claimed privacy guarantee.
\end{proof}

\begin{lemma}[Efficiency] \label{lem:dpmyersen-efficiency}
Suppose that $\countingalg$ runs in $\atime_c(n)$ time and uses $\aspace_c(n)$ space, and $\densealg$ runs in $\atime_d(n)$ time and uses $\aspace_d(n)$ space. Furthermore, suppose that the output stream $O_{\leq n}$ has size at most $K$. Then, \Cref{alg:meyerson-main} runs in $O\Paren{nd K + \atime_c(n) + \Psi \cdot \atime_d(\Psi \cdot \tau)}$ time and uses $O\Paren{\Psi \cdot \Paren{\tau + \Psi \cdot \Upsilon} \cdot d + \aspace_c(n) + \aspace_d(\Psi \cdot \tau)}$ space.
\end{lemma}

\begin{proof}
For the running time, observe that all calculations outside the calls to $\countingalg,\densealg$ takes $O(d)$ time per point except for the computation of the closest point (Line~\ref{line:closest-point}) which may take up to $O(K \cdot d)$ time. Furthermore, we only invoke $\densealg$ at most $O(\Psi)$ times, and each time the input $S_t$ has size at most $\Psi \cdot \tau$. The total running time thus follows.

As for the space complexity, outside of $\countingalg,\densealg$, the only thing needed to be kept by our algorithm are $S_t, R_t$. The former has size at most $\Psi \cdot \tau$ and the latter has size at most $\Psi^2 \cdot \Upsilon$ (due to the size blow-up guarantee of $\densealg$). Finally, note that each element in $S_t, R_t$ are points in $\R^d$ so it requires $O(d)$ space.
\end{proof}

\subsection{Combined Prefix-Approximation Guarantee via Guessing}

By using a randomized guessing strategy for the estimated optimum cost, we can eliminate the dependence on $\estopt$ in the error bounds of the Private Meyerson's Sketch. This allows us to strictly parameterize the capacity bound $\Psi$ and the density threshold $\tau$ in terms of $k, d, n$.

\begin{algorithm}[h]
\caption{Random-Guess DP Meyerson's Sketch $\mathcal{A}_{\text{guess-bicri}}$}
\label{alg:guess-bicriteria}
\begin{algorithmic}[1]
\REQUIRE Stream $X = (x_1, \dots, x_n)$ of points in a metric space $\mathcal{M}=(U,d)$ with diameter $\le 1$.
\STATE \textbf{Parameters:} Target centers $k$, Privacy budgets $(\eps, \delta)$.
\STATE Define the set of possible guesses: $\mathcal{G} = \{2^{\lceil \log_2 n \rceil - 1}, 2^{\lceil \log_2 n \rceil-2}, \dots, 2^{-\lceil \log_2 n \rceil}\}$
\STATE Sample an estimated optimum $\estopt$ uniformly at random from $\mathcal{G}$.
\STATE Set the target sketch capacity: $\Psi = \Theta(k \log n)$
\STATE Set the density threshold $\tau$ depending on the privacy model:
\begin{itemize}
    \item For pure-DP: $\tau = \Theta\Paren{\frac{d \cdot k^2}{\eps} \polylog{\frac{n d}{\eps}} }$
    \item For approximate-DP: $\tau = \Theta\Paren{\frac{\sqrt{d} \cdot k^{1.5}}{\eps} \polylog{\frac{n d}{\eps \delta}} }$
\end{itemize}
\STATE Instantiate $\countingalg$ as the $(\eps/2, \delta/2)$-DP continual counting algorithm.
\STATE Instantiate $\densealg$ as the $(\frac{\eps}{4\Psi}, \frac{\delta}{8\Psi e^{\eps/(2\Psi)}})$-DP dense ball algorithm.
\STATE Run $\meyersonalg(\estopt, \Psi)$ on $X$ using the configured parameters.
\RETURN Output stream $(O_1, \dots, O_n)$ from the sketch.
\end{algorithmic}
\end{algorithm}

\begin{lemma}[DP Prefix-Approximation] \label{lem:combined-bicriteria-guess}
Given a stream $X = (x_1, \dots, x_n)$ of $n$ points and a target number of centers $k$, \Cref{alg:guess-bicriteria} is $(\eps, \delta)$-DP. Furthermore, it outputs a multiset-stream $(O_1, \dots, O_n)$ satisfying:
\begin{enumerate}
    \item \textbf{Accuracy:} For a fixed $t \in [n]$, with probability at least $\Omega\Paren{\frac{1}{\log n}}$, $(O_1, \dots, O_t)$ is an $(O(1), \beta)$-prefix-approximation of $(x_1, \dots, x_t)$, where the additive error $\beta$ is bounded by:
    \begin{itemize}
        \item $\beta = O\Paren{\frac{d \cdot k^3}{\eps} \polylog{\frac{n d}{\eps}} }$ for pure-DP.
        \item $\beta = O\Paren{\frac{\sqrt{d} \cdot k^{2.5}}{\eps} \polylog{\frac{n d}{\eps \delta}} }$ for approximate-DP.
    \end{itemize}
    \item \textbf{Size Bound:} The total number of output centers is bounded by $|O_{\le n}| \le O(k^2 \log^2 n)$.
\end{enumerate}
The algorithm processes the stream in total time $\tO{ndk^2} + \Lambda^{O(1)}$ and uses space $\Lambda^{O(1)}$, where
\begin{itemize}
\item $\Lambda = \frac{kd \log n}{\eps}$ for pure-DP.
\item $\Lambda = \frac{kd \log(n/\delta)}{\eps}$ for approximate-DP.
\end{itemize}
\end{lemma}

\begin{proof}
The privacy guarantee of \Cref{alg:guess-bicriteria} immediately follows from that of $\meyersonalg$ (\Cref{lem:dpmyersen-privacy}).

\vspace{2mm}\noindent\textbf{Accuracy Guarantee (Conditioning on a Correct Guess).}
Fix $t \in [n]$.
Let $\OPT := \OPT^{k,p}_{X_{\leq t}}$ denote the true optimal cost. Since the metric space has diameter $\le 1$ and there are $t \leq n$ points, $\OPT \le n$. Let us define the ``correct guess'' $g^* \in \mathcal{G}$ as follows:
\begin{itemize}
    \item If $\OPT > 1/n$, let $g^*$ be the unique power of 2 in $\mathcal{G}$ such that $g^* < \OPT \le 2g^*$.
    \item If $\OPT \le 1/n$, let $g^* = 2^{-\lceil \log_2 n \rceil} \le 1/n$.
\end{itemize}
Since $\mathcal{G}$ contains exactly $2\lceil \log_2 n \rceil$ elements, the algorithm selects $\estopt = g^*$ with probability exactly $\frac{1}{2\lceil \log_2 n \rceil} = \Omega\Paren{\frac{1}{\log n}}$.  Conditioned on the event that $\estopt = g^*$, we now verify that the parameters assigned strictly satisfy the prerequisites for the prefix-approximation guarantee of the private Meyerson's sketch in \Cref{lem:bicriteria-approx}:
\begin{enumerate}
    \item \textbf{Capacity $\Psi$:} The prerequisite is $\Psi \ge D \cdot kL \Paren{1 + \frac{\OPT}{\estopt}}$ where $L = \Theta\Paren{ \log_2(n/\estopt) }$. Because $g^*$ is the correct guess, we have $\frac{\OPT}{\estopt} \le \frac{\OPT}{g^*} \le 2$ (or $\frac{\OPT}{1/n} \le 1$). 
    Since $\estopt \ge 1/n$, $L \le  \Theta(\log n)$. 
    Thus, it suffices to choose $\Psi = \Theta(k \log n)$ where $\Theta(\cdot)$ hides a sufficiently large constant.
    
    \item \textbf{Threshold $\tau$:} We require both $\countingalg$ and $\densealg$ to succeed with failure probability $\gamma \le 1/(4n)$.
    \begin{itemize}
        \item $\countingalg$ requires an additive error of $\le \tau/2$. Given $\eps_c = \eps/2$, by \Cref{thm:cont-counting}, it suffices to satisfy $\tau \ge \Omega\Paren{\frac{1}{\eps} \cdot \polylog{n}}$. 
        \item $\densealg$ processes datasets of size up to $N = \Psi \cdot \tau$.
        \begin{itemize}
            \item \textbf{For pure-DP:} By \Cref{thm:dense-ball-main}, the subroutine requires $\tau \ge \Theta\Paren{\Paren{\frac{d \cdot (\Psi \tau)}{\eps_d}}^{1/2} \cdot \polylog{\frac{n d}{\eps_d}}}$. Since $\eps_d = \eps / (4\Psi)$, this resolves to:
            \begin{align*}
                \tau^2 \ge \Theta\Paren{\frac{d \Psi \tau}{\eps / \Psi} \cdot \polylog{\frac{n d}{\eps}} } \iff \tau \ge \Omega\Paren{\frac{d \Psi^2}{\eps} \cdot \polylog{\frac{n d}{\eps}} }
            \end{align*}
            Substituting $\Psi = \Theta(k \log n)$, it suffices to take $\tau = \Theta\Paren{\frac{d k^2}{\eps} \cdot \polylog{\frac{n d}{\eps}} }$.
            
            \item \textbf{For approximate-DP:} By \Cref{thm:dense-ball-main}, this requires $\tau \ge \Theta\Paren{\Paren{\frac{d \cdot (\Psi \tau)}{\eps_d^2}}^{1/3} \cdot \polylog{\frac{n d}{\eps_d \delta}}}$. This resolves to:
            \begin{align*}
                \tau^3 \ge \Theta\Paren{\frac{d \Psi \tau}{(\eps/\Psi)^2} \cdot \polylog{\frac{n d}{\eps \delta}} } &\iff \tau^2 \ge \Theta\Paren{\frac{d \Psi^3}{\eps^2} \cdot \polylog{\frac{n d}{\eps \delta}} } \\ &\iff \tau \ge \Theta\Paren{\frac{\sqrt{d} \Psi^{1.5}}{\eps} \cdot \polylog{\frac{n d}{\eps \delta}} }
            \end{align*}
            Substituting $\Psi = \Theta(k \log n)$, it suffices to take $\tau = \Theta\Paren{\frac{\sqrt{d} k^{1.5}}{\eps} \cdot \polylog{\frac{n d}{\eps \delta}} }$.
        \end{itemize}
    \end{itemize}
\end{enumerate}

Since all parameter prerequisites are met, the private Meyerson's sketch guarantees that with probability at least $1/2$, the output $(O_1, \dots, O_t)$ is an $(O(1), \beta)$-prefix-approximation of $(x_1, \dots, x_t)$, where the additive error bound is given by:
\begin{align*}
    \beta_{bicri} = O(\estopt + \tau k L).
\end{align*}
Recall that $\estopt \le \max(\OPT, 1/n)$; thus, the term $O(\estopt)$ scales as $O(\OPT) + O(1/n)$. The $O(\OPT)$ component acts as a constant multiplier absorbing directly into the $O(1)$-multiplicative approximation factor. The $O(1/n)$ term is strictly subsumed by the dominant term $O(\tau k L)$. Thus, the final additive error is purely $\beta = O(\tau k L) = O(\tau k \log n)$. 
Substituting our values for $\tau$:
\begin{itemize}
    \item \textbf{Pure-DP:} $\beta = O\Paren{ \frac{d k^2}{\eps} \cdot \polylog{\frac{n d}{\eps}} \cdot k \log n } = O\Paren{ \frac{d k^3}{\eps} \cdot \polylog{\frac{n d}{\eps}} }$.
    \item \textbf{Approximate-DP:} $\beta = O\Paren{ \frac{\sqrt{d} k^{1.5}}{\eps} \cdot \polylog{\frac{n d}{\eps \delta}} \cdot k \log n } = O\Paren{ \frac{\sqrt{d} k^{2.5}}{\eps} \cdot \polylog{\frac{n d}{\eps \delta}} }$.
\end{itemize}
Conditioned on guessing $g^*$, the algorithm succeeds with probability $1/2$. Thus, the overall unconditional success probability is strictly $\Omega(1/\log n) \cdot \frac{1}{2} = \Omega\Paren{\frac{1}{\log n}}$.

\vspace{2mm}\noindent\textbf{Size Bound and Complexity.}
By the constraints of the private Meyerson's sketch, the total number of output centers $K$ is bounded by $\Psi^2 \cdot \Upsilon$. Because $\Upsilon = O(1)$ for the 1-Cluster to dense ball reduction, we have $K \leq O(\Psi^2) = O(k^2 \log^2 n)$.

The time and space complexity follows by plugging the above parameters into \Cref{lem:dpmyersen-efficiency} with $\atime_c(n) = O(n \log n), \aspace_c(n) = O(n)$ (from \Cref{thm:cont-counting}) and $\atime_d(N) = \aspace_c(N) = N^{O(1)}$ (from \Cref{thm:dense-ball-main}) with $N \leq \Psi \cdot \tau$.
\end{proof}

\subsection{Boosting Probability via Repetition}

Since \Cref{lem:combined-bicriteria-guess} only yields a low-probability guarantee for a single $t \in [n]$, we need to boost its success probability. Fortunately, this can be easily accomplished by repeatedly running the algorithm multiple times. We start by describing the generic reduction below.

\begin{lemma} \label{lem:low-to-high-prob}
Suppose that there is an $(\teps, \tdelta)$-DP algorithm that satisfies the following:
\begin{itemize}
\item Its output stream size is bounded by $B$.
\item Its space and time complexity are $\aspace(n)$ and $\atime(n)$.
\item For each fixed $t \in [n]$, the output stream $(O_1, \dots, O_t)$ is an $(\alpha, \beta)$-prefix-approximation of the input stream $(x_1, \dots, x_t)$ with probability at least $q$.
\end{itemize}
Then, for some $R = O\Paren{\frac{\log(n/\gamma)}{q}}$ there is an $\Paren{R \cdot \teps, R \cdot \tdelta}$-DP algorithm that satisfies the following:
\begin{itemize}
\item Its output stream size is bounded by $R \cdot B$.
\item Its space and time complexity are $O(R) \cdot \aspace(n)$ and $O(R) \cdot \atime(n)$.
\item With probability $1 - \gamma$, for all $t \in [n]$, the output stream $(O_1, \dots, O_t)$ is an $(\alpha, \beta)$-prefix-approximation of the input stream $(x_1, \dots, x_t)$.
\end{itemize}
\end{lemma}

\begin{proof}
This follows by simply running the algorithm $R = \left\lceil \frac{\ln(n/\gamma)}{q} \right\rceil$ times, and outputting the union of all the output. The output size bound, and space and time complexity obviously grows by a factor of $R$. As for the $(\alpha, \beta)$-prefix-approximation guarantee, observe that it holds if at least one of the repetition satisfies the $(\alpha, \beta)$-prefix-approximation. Thus, by a union bound, the probability that it satisfies $(\alpha, \beta)$-prefix-approximation for all $t \in [n]$ is at least $1 - n \cdot (1 - q)^R \geq 1 - \gamma$. 
\end{proof}

Combining \Cref{lem:low-to-high-prob} with \Cref{lem:combined-bicriteria-guess} (where $q = \Omega(1 / \log n)$), we immediately get the following:

\begin{theorem}[High-Probability DP Prefix-Approximation] \label{thm:combined-bicriteria-high-prob}
Given a stream $X = (x_1, \dots, x_n)$ of $n$ points and a target number of centers $k$, there is an $(\eps, \delta)$-DP algorithm whose multiset-stream $(O_1, \dots, O_n)$ satisfies:
\begin{enumerate}
    \item \textbf{Accuracy:} With probability at least $1 - \gamma$, for all $t \in [n]$, the output stream $(O_1, \dots, O_t)$ is an $(O(1), \beta)$-prefix-approximation of the input stream $(x_1, \dots, x_t)$, where the additive error $\beta$ is bounded by:
    \begin{itemize}
        \item $\beta = O\Paren{\frac{d \cdot k^3}{\eps} \polylog{\frac{n d}{\eps \gamma}} }$ for pure-DP.
        \item $\beta = O\Paren{\frac{\sqrt{d} \cdot k^{2.5}}{\eps} \polylog{\frac{n d}{\eps \delta \gamma}} }$ for approximate-DP.
    \end{itemize}
    \item \textbf{Size Bound:} The total number of output centers is bounded by $|O_{\le n}| \le O(k^2 \cdot \polylog{n/\gamma})$.
\end{enumerate}
The algorithm processes the stream in total time $\tO{ndk^2 \cdot \log(1/\gamma)} + \Lambda^{O(1)}$ and uses space $\Lambda^{O(1)}$, where
\begin{itemize}
\item $\Lambda = \frac{kd \log(n/\gamma)}{\eps}$ for pure-DP.
\item $\Lambda = \frac{kd \log\Paren{ \frac{n}{\gamma\delta} } }{\eps}$ for 
approximate-DP.
\end{itemize}
\end{theorem}
\section{From Prefix-Approximation to DP (Semi-)Coreset: Proof of \texorpdfstring{\Cref{thm:main-semi-coreset}}{Theorem~\ref{thm:main-semi-coreset}}} \label{subsec:bicriteria-to-coreset}

In this section, we present our algorithm for constructing semi-coreset and prove our main theorem (\Cref{thm:main-semi-coreset}). As explained in \Cref{sec:dp-meyerson-main}, we simply run the online prefix-approximation algorithm and then snap $x_t$ to the closest centers among those produced by the algorithm. We use DP continual counting algorithm to estimate the weight of each center.

\begin{algorithm}[h]
\caption{Private Semi-Coreset $\dpcoreset$}
\label{alg:coreset-main}
\begin{algorithmic}[1]
\REQUIRE Stream $X = (x_1, \dots, x_n)$
\item[\textbf{Parameters:}] Online Prefix-Approximation Algorithm $\bicrialg$; Continual Counting Algorithm $\countingalg$
\FOR{$t = 1, \dots, n$}
\STATE \algcomment{Do the following upon receiving $x_t$}
\STATE Feed $x_t$ to $\bicrialg$ and let $y_t \subseteq U$ be the output
\FOR{$u \in y_t \setminus Y_{\leq t - 1}$}
\STATE $\countingalg^u \gets$ start $\countingalg$ instance for $u$ 
\STATE $e^u_{t-1} \gets 0$
\ENDFOR
\STATE $\tx_t \gets \emptyset$ \hfill \algcomment{Output in this time step}
\IF{$x_t \ne \perp$}
\STATE $f_t \gets$ closest point to $x_t$ in $Y_{\leq t}$ \hfill \algcomment{Ties  broken arbitrarily}
\ELSE
\STATE $f_t \gets \perp$
\ENDIF
\FOR{$u \in Y_{\leq t}$}
\IF{$u = f_t$}
\STATE Feed 1 to $\countingalg^u$ and let $e^u_t$ be the output
\ELSE
\STATE Feed 0 to $\countingalg^u$ and let $e^u_t$ be the output
\ENDIF
\STATE Add $(e^u_t - e^u_{t - 1})$ copies of $u$ to $\tx_t$
\ENDFOR
\RETURN $\tx_t$
\ENDFOR
\end{algorithmic}
\end{algorithm}

\subsection{Properties of the DP Semi-Coreset Algorithm: Privacy and Efficiency}

The privacy and efficiency guarantee of \Cref{alg:coreset-main} are fairly simple to derive and are stated more formally below.

\begin{lemma}[Privacy] \label{lem:coreset-privacy}
If $\bicrialg$ is $(\eps_b, \delta_b)$-DP and $\countingalg$ is $\eps_c$-DP, then \Cref{alg:coreset-main} is $(\eps_b + \eps_c, \delta_b)$-DP.
\end{lemma}

\begin{proof}
Fix an adversary $\mathcal{B}$. Recall that we use $i^\ast$ to denote the challenge index. 

First, it will be more convenient to consider a modified algorithm $\widehat{\mathcal{A}}$, which is exactly the same as \Cref{alg:coreset-main}, except it releases $(\tx_t, y_t)$ at each step (instead of just $\tx_t$). By
\Cref{lem:postproc}, it suffice to prove that $\widehat{\mathcal{A}}$ is $(\eps_b + \eps_c, \delta_b)$-DP.


We view  $\widehat{\cA}$ as the concurrent composition of $\bicrialg$ and another mechanism $\widehat{\cA}'$ that handles the counters $\{\countingalg^u\}_{u\in Y_{\le n}}$. 
\begin{itemize}
\item For $\bicrialg$, the submitted stream differs only at $i^\ast$. Thus, this is $(\eps_b,\delta_b)$-DP.
\item For $\widehat{\cA}$, conditioned on a transcript prefix $\pi_{t-1} = ((\tx_1, y_1), \ldots, (\tx_{t-1}, y_{t-1}))$ and the output $y_t$ from $\bicrialg$, $Y_{\le t} = \bigcup_{t' \leq t} y_{t'}$ agrees in the two runs. Thus, for $t\ne i^\ast$ the submitted inputs to $\{\countingalg^u\}_{u\in Y_{\le t}}$ agree. For $t = i^*$, assume w.l.o.g. that $x_{i^\ast}^1=\perp$. For the run $b = 1$, all the inputs to the counters are 0. Meanwhile, for $b = 0$, all the inputs remain 0, except for $\countingalg^{f_t^0}$. Hence, we may apply \Cref{lem:parallel} to conclude that this is $\eps_c$-DP. 
\end{itemize}
Finally, applying Lemma~\ref{lem:concurrent} over these two components gives the claimed privacy guarantee.
\end{proof}

\begin{lemma}[Efficiency] \label{lem:coreset-efficiency}
Let $K$ be the maximum number of centers produced by $\bicrialg$ (i.e., $|Y_{\leq n}| \le K$). If $\bicrialg$ runs in time $\mathcal{T}_b(n)$ and space $\mathcal{S}_b(n)$, and each instance of $\countingalg$ runs in time $\mathcal{T}_c(n)$ and space $\mathcal{S}_c(n)$, then Algorithm \ref{alg:coreset-main} runs in time $O(\mathcal{T}_b(n) + n K d + K \cdot \mathcal{T}_c(n))$ and uses space $O(\mathcal{S}_b(n) + K \cdot \mathcal{S}_c(n) + K d)$.
\end{lemma}

\begin{proof}
Algorithm \ref{alg:coreset-main} runs $\bicrialg$ across the stream, taking $\mathcal{T}_b(n)$ time and $\mathcal{S}_b(n)$ space. It maintains a set $S$ of at most $K$ points in $\mathbb{R}^d$, requiring $O(Kd)$ space. At each of the $n$ time steps, it computes the distance from $x_t$ to all points currently in $S$, taking $O(Kd)$ time per step. It then feeds an indicator bit to each of the $|Y_{\leq t}| \le K$ active instances of $\countingalg$. Running $K$ counters takes $O(K \cdot \mathcal{T}_c(n))$ time and $O(K \cdot \mathcal{S}_c(n))$ space over the stream. Summing these bounds yields the stated efficiency.
\end{proof}

\subsection{Semi-Coreset Guarantee}

Next, we prove that the output of \Cref{alg:coreset-main} provides an output that is a semi-coreset of the input stream. We state the lemma with generic parameters below.

\begin{lemma}[Semi-Coreset Guarantee] \label{lem:coreset-guarantee}
Suppose that with high probability, $\bicrialg$ provides an $(\alpha_1, \beta_1)$-prefix-approximation for all time steps $t \in [n]$, and $\countingalg$ is $\beta_2$-accurate for all instances and all time steps. Then, with high probability, the output multiset stream $\tX = (\tx_1, \dots, \tx_n)$ forms a $(\kappa, \eta_1, \eta_2)$-semi-coreset of the input $X$ where $\kappa = O(1 + \alpha_1)$, $\eta_1 = O(\alpha_1)$, and $\eta_2 = O(\beta_1 + K \beta_2)$, with $K = |S|$ representing the total number of centers generated up to time $t$.
\end{lemma}

\begin{proof}
Fix a time step $t \in [n]$. 
Algorithm \ref{alg:coreset-main} constructs the multiset $\tX_{\le t}$, which consists of points from $Y_{\leq t}$ where the multiplicity of each center $u \in S$ is $e^u_t$, the estimate provided by the respective instance $\countingalg^u$. 
    
We define a natural mapping $\Psi: X_{\le t} \to Y_{\le t}$ that maps each point $x_i$ ($i \leq t$) directly to the closest available center $f_i \in Y_{\le i}$ at the time of its arrival: $\Psi(x_i) = f_i$.
    
Let $c^u_t = \sum_{i=1}^t \mathbf{1}[f_i = u]$ be the exact number of points in $X_{\le t}$ that were mapped to center $u$. The generalized transport cost from the original points $X_{\le t}$ to the estimated coreset $U_{\le t}$ under mapping $\Psi$ is exactly:
\begin{align*}
    mt(\Psi, X_{\le t}, \tX_{\le t}) = \sum_{i=1}^t \dist(x_i, f_i)^p + \sum_{u \in Y_{\le t}} |c^u_t - e^u_t|.
\end{align*}
We bound the two terms separately:
\begin{enumerate}
    \item \textbf{Distance Term:} Since $f_i$ is chosen as the closest point to $x_i$ from the active centers $Y_{\le i}$ output by $\bicrialg$, the distance $\dist(x_i, f_i)^p$ is exactly $d(x_i, Y_{\le i})^p$. By the prefix-approximation guarantee of $\bicrialg$, we sum this over time to get:
    \begin{align*}
        \sum_{i=1}^t \dist(x_i, f_i)^p \leq \alpha_1 \cdot \OPT^{k,p}_{X_{\le t}} + \beta_1.
    \end{align*}
    \item \textbf{Weight Difference Term:} The true frequency of assignment is $c^u_t$, while the constructed coreset uses the estimated frequency $e^u_t$. By the assumed accuracy guarantee of $\countingalg$, we have $|c^u_t - e^u_t| \leq \beta_2$ for all $u \in Y_{\le t}$. Summing this bounded error over all $K = |Y_{\le t}|$ active centers yields:
    \begin{align*}
        \sum_{u \in Y_{\le t}} |c^u_t - e^u_t| \leq K \beta_2.
    \end{align*}
\end{enumerate}
Combining these bounds, the total transport cost strictly satisfies:
\begin{align*}
    mt(\Psi, X_{\le t}, U_{\le t}) \leq \alpha_1 \OPT^{k,p}_{X_{\le t}} + (\beta_1 + K \beta_2).
\end{align*}
Because the metric space has diameter $\le 1$, we can directly apply Lemma \ref{lem:transport_coreset} to this result. Substituting $\alpha = \alpha_1$ and $\beta = \beta_1 + K \beta_2$, it follows immediately that $U_{\le t}$ is a valid $(\kappa, \eta_1, \eta_2)$-semi-coreset for $X_{\le t}$ with parameters $\kappa = O(1 + \alpha_1)$, $\eta_1 = O(\alpha_1)$, and $\eta_2 = O(\beta_1 + K \beta_2)$. 
\end{proof}

\subsection{Putting Things Together: Proof of \texorpdfstring{\Cref{thm:main-semi-coreset}}{Theorem~\ref{thm:main-semi-coreset}}}

\begin{proof}[Proof of \Cref{thm:main-semi-coreset}]
The algorithm is simply to run \Cref{alg:coreset-main} using $\bicrialg$ and $\countingalg$ from \Cref{thm:combined-bicriteria-high-prob} and \Cref{thm:cont-counting} respectively, where the privacy parameters are set to $\eps_b = \eps_c = \frac{\eps}{2}$ and $\delta_b = \delta$. The privacy guarantee immediately follows from \Cref{lem:coreset-privacy}.

Recall also that the size bound by \Cref{thm:combined-bicriteria-high-prob} is $K = O(k^2 \cdot \polylog{n/\gamma})$. Plugging this (and the complexity bounds from \Cref{thm:combined-bicriteria-high-prob,thm:cont-counting}) into \Cref{lem:coreset-efficiency} yields the claimed time and space complexity.

Finally, plugging $K = O(k^2 \cdot \polylog{n/\gamma})$ and the accuracy guarantees from \Cref{thm:combined-bicriteria-high-prob,thm:cont-counting} into \Cref{lem:coreset-guarantee} yields the semi-coreset guarantee. (Note that the error term is dominated by $\beta_1$ in both pure-DP and approximate-DP.) Finally, the number of distinct points in the output is at most $K$.
\end{proof}

\end{document}